\documentclass[12pt]{article}

\usepackage{fancyhdr,amsmath,amsthm,amssymb,bm} 
\usepackage{graphicx,caption}
\usepackage{float,url}
\usepackage{bm}
\usepackage{bbm}
\usepackage{dsfont}
\usepackage[top=1.1in, bottom=1.1in, left=1.1in, right=1.1in]{geometry}
\usepackage{color}
\usepackage{diagbox}
\usepackage{authblk}
\usepackage{algorithm}
\usepackage{algorithmic}
\usepackage{mathtools}

\usepackage[colorlinks=true,
            linkcolor=blue,
            urlcolor=blue,
            citecolor=blue]{hyperref}
            
\DeclareCaptionFormat{myformat}{#3}

\usepackage{mkolar_definitions}

\let\hat\widehat
\let\tilde\widetilde

\graphicspath{{fig/}}

\begin{document}

\author{Ming Yu\thanks{Booth School of Business, The University of Chicago. Email: \href{mailto:ming93@uchicago.edu}{ming93@uchicago.edu} } , 
Varun Gupta\thanks{Booth School of Business, The University of Chicago. Email: \href{mailto:varun.gupta@chicagobooth.edu}{varun.gupta@chicagobooth.edu}} , 
and Mladen Kolar\thanks{Booth School of Business, The University of Chicago. Email: \href{mailto:mladen.kolar@chicagobooth.edu}{mladen.kolar@chicagobooth.edu}}}

\title{Constrained High Dimensional Statistical Inference}
\date{}
\maketitle

\begin{abstract} \label{sec:abstract}

In typical high dimensional statistical inference problems, confidence intervals and hypothesis tests are performed for a low dimensional subset of model parameters under the assumption that the parameters of interest are unconstrained. However, in many problems, there are natural constraints on model parameters and one is interested in whether the parameters are on the boundary of the constraint or not. e.g. non-negativity constraints for transmission rates in network diffusion. In this paper, we provide algorithms to solve this problem of hypothesis testing in high-dimensional statistical models under constrained parameter space. We show that following our testing procedure we are able to get asymptotic designed Type I error under the null. Numerical experiments demonstrate that our algorithm has greater power than the standard algorithms where the constraints are ignored.
We demonstrate the effectiveness of our algorithms on two real datasets where we have {\emph{intrinsic}} constraints on the parameters.

%

\end{abstract}

\section{Introduction}\label{sec:introduction}

Statistical estimation of high dimensional problems has been attracting more and more attention due to the abundance of such data in many emerging fields such as genetic studies, social network analysis, etc.
High dimensional geometry is inherently different from low-dimensional geometry. 
As an example, for linear regression, in low dimensions the Ordinary Least Square (OLS) estimator allows for constructing confidence intervals and hypothesis tests for the true coefficient. In high dimensional models OLS is ill-conditioned so instead we have to solve for penalized estimators like LASSO. 
In low dimensions we can test for hypothesis such as $H_0 : \alpha^* = 0$ by partial likelihood function while in high dimensions this also fails, due to the large amount of nuisance parameters.

In this paper we consider a hypothesis testing problem in a high dimensional model under constrained parameter space. 
For many problems, before analyzing data and fitting models we might already know some constraints on the parameters. 
This can also be viewed as prior information on the parameters.
For example in isotonic regression \cite{best1990active, yang2019contraction, dai2019bias} we have a constraint that the variables are non-decreasing; in non-negative least square problem \cite{slawski2013non} we have a constraint that the coefficients are non-negative. 
in real-world reinforcement learning applications, we need to take into consideration the safety of the agent \cite{berkenkamp2017safe, wen2018constrained, yu2019convergent}. 
Also, in Gaussian process, it is sometimes assumed that the parameters satisfy some linear inequality constraints \cite{lopez2019gaussian}.

With this additional information the statistical inference and hypothesis testing for the parameters may be different. For example, consider a simple model: $X \sim N(\mu,1)$. In general if we want to test whether $\mu$ is 0 or not, i.e. test for $H_0 : \mu=0$ versus $H_A : \mu \neq 0$, we will reject $H_0$ if the absolute value of the mean $|\bar x|$ is relatively large. However, if we have the constraint that $\mu \geq 0$, then we are testing $H_0 : \mu=0$ versus $H_A : \mu > 0$, and we reject $H_0$ only when $\bar x$ is relatively large.

When we have constraints on parameters, a natural question we want to answer is
whether the parameter lies on the boundary or is away from the boundary, since these two cases are usually very different.
For example for nonnegativity constraint, we want to know whether the parameter is exactly zero or strictly positive; for monotonic constraint we want to know whether the two variables are equal or one is strictly greater than the other. 


In this paper we perform statistical inference (hypothesis testing) for low dimensional parameters in a high dimensional model under cone constraint. Denote the parameter $\bm \beta = ( \bm \alpha, \bm \theta)$, where $\bm \alpha$ is the low dimensional parameter of interest and $\bm \theta$ denotes nuisance parameters. Denote the constraint set $C$ as a closed and convex cone, and let $M \in C$ be a linear space in $C$. 
In most of the cases $C$ is a polyhedron and the linear space $M$ denotes the (subset of) the boundary set of $C$. 
In this paper we want to test 
\begin{equation}
H_0 : \bm \alpha \in M \,\,\, \text{versus} \,\,\, H_A : \bm \alpha \in C\backslash M,
\end{equation}
where we have the constraint $\bm \alpha \in C$. 
We develop an algorithm for this constrained testing problem in high dimensional models. Following our procedure we show that the hypothesis test method we propose has asymptotic designed Type I error, and it has much greater power than when the constraints are ignored.
\subsection{Related Work}\label{sec:related_work}

\paragraph{High-dimensional inference without constraint.} 
There is a vast literature on performing statistical inference for high dimensional models and here we provide a brief overview. Early work \cite{knight2000asymptotics} shows that the limiting distribution of LASSO estimator is not normal even in low dimensions. More recently, 
several approaches have been proposed to obtain asymptotic distribution on low dimensional parameters in high dimensional linear model, mostly by approximating the inverse of the Gram matrix. 
\cite{zhang2014confidence} gives confidence intervals for low dimensional parameters in a high dimensional linear model using low dimensional projection estimator (LDPE). 
\cite{javanmard2014confidence} provides asymptotic confidence interval of LASSO estimator for high dimensional linear regression by introducing the debiasing method. 
\cite{van2014asymptotically} further extends their work to a more general setting, including Generalized Linear Model and other nonlinear models. 
\cite{ning2017general} deals with general model on Hessian matrix with Dantzig type estimator. 
Related works also include \cite{dezeure2017high, zhang2017simultaneous} for simultaneous inference, \cite{belloni2014inference} for double selection method, \cite{ren2015asymptotic, yu2016statistical, yu2019simultaneous} for graphical model, \cite{taylor2014exact, yang2016selective, lee2016exact} for post selective inference, \cite{li2019statistical, cao2019estimation} for for synthetic control, \cite{song2019convex} for noisy labels, etc.




\paragraph{Low-dimensional constrained inference.} 
The literature on constrained testing dates back to \cite{chernoff1954distribution}, where the authors prove the asymptotic distribution of the likelihood ratio (LR) test statistic for constrained testing to be weighted Chi-square. 
\cite{perlman1969one} further considers testing with unknown covariance matrix, and gives sharp upper and lower bounds for the weights. 
\cite{gourieroux1982likelihood} introduces the test statistics for likelihood ratio test, Wald test, and Kuhn-Tucker test with inequality constraint in linear model, and proves the equivalence of these three tests. 
\cite{king1986joint} proposes one-sided $t$-test when the coefficients' signs are known. 
\cite{rogers1986modified} introduces a modified Lagrange multiplier test for testing one-sided problem. 
\cite{kodde1986wald} proposes Wald test for jointly testing equality and inequality constraints on the parameters.
\cite{wolak1989testing} develops asymptotically equivalent tests under linear inequality restrictions for linear models. \cite{king1997locally} introduces a locally most mean powerful (LMMP) test. 
\cite{andrews1998hypothesis} introduces directed tests, which is optimal in terms of power. 
\cite{bloch2001one} introduces multiple-endpoint testing in clinical trials. \cite{hall2001order} provides Order-Restricted Score Tests for generalized linear and nonlinear mixed models. \cite{andrews2001testing} proposes test when nuisance parameters appear under the alternative hypothesis, but not under the null. 
\cite{perlman2006some} gives improved LRT and UIT test. 
More recently, \cite{lu2013halfline} has discussed halfline test for inequality constraints. \cite{susko2013likelihood} gives conservative likelihood ratio test using data-dependent degree of freedom.  
\cite{zhu2014testing} gives Wald test under inequality constraint in linear model with spherically symmetric disturbance. 
\cite{lu2015extended} proposes an extended MaxT test and gets the power improvement. 
However, all these existing results are for low dimensional models.

In terms of statistical inference, our work is most related to \cite{ning2017general}, where the authors establish inference for high dimensional models using decorrelation method. We will review this method in Section \ref{sec:procedure}. For constrained testing, our work is most related to \cite{shapiro1988towards} where the authors introduce and discuss Chi-bar-squared statistic, and \cite{silvapulle1995score} and \cite{molenberghs2007likelihood} which form the one sided test to test whether a parameter is zero or strictly positive. 
Recent works \cite{javanmard2017flexible, wei2019geometry} consider hypothesis testing on whether the parameters lie in some convex cone. This is still different from our setting where we know the parameters lie in the convex cone and the goal is to test whether they lie on the boundary of the cone.

\subsection{Organization of the Paper}

In Section \ref{sec:procedure} we give the detailed procedure for our algorithm. Section \ref{sec:main_theorem} gives assumptions under which our method is valid, and states our main theorem. 
Sections \ref{sec:synthetic_data} and \ref{sec:real_data} present experimental results on synthetic datasets and real world datasets, respectively. We conclude in Section \ref{sec:conclusion}.

\section{Algorithm}\label{sec:procedure}

In this section we describe our main algorithm.
Consider a high dimensional statistical model with parameters $\bm \beta \in \mathbb R^p$ and the partition $\bm \beta = ( \bm\alpha , \bm \theta)$, where $\bm\alpha$ is $d$ dimensional parameter of interest, and $\bm \theta$ is a $p-d$ dimensional nuisance parameter with $d \ll p$. We write $\bm \alpha = (\alpha_1, ..., \alpha_d)$, and the true parameter as $\bm \beta^* = (\bm\alpha^* , \bm \theta^*)$ with $\| \bm \beta^* \|_0 = s$. Moreover, we have the constraint $\bm \alpha^* \in C$ where $C$ is a closed and convex cone. Let $M \in C$ be a linear space in $C$. 
In most of the cases $C$ is an polyhedron and the linear space $M$ denotes the (subset of) the boundary set of $C$. 
The hypothesis we want to test is 
\begin{equation}
\label{eq:main_hypothesis}
H_0 : \bm\alpha^* \in M \,\,\, \text{versus} \,\,\, H_A : \bm\alpha^* \in C\backslash M, 
\end{equation}
i.e. we want to test whether $\bm\alpha^*$ lies on the boundary of $C$, or is a strict interior point of $C$ in at least one direction. 
For example, with nonnegativity constraint we have $C = \RR^d_+ = \{ \bm\alpha: \bm\alpha \geq \bm 0\}$ and $M = \{ \bm\alpha: \bm\alpha = \bm 0\}$. The hypothesis we want to test is 
\begin{equation}
H_0: \bm\alpha^* = \bm 0 \quad\text{versus}\quad H_A: \exists j\in \{1,...,d\} \,\,\, \text{s.t.} \, \,\,\alpha_j^* > 0. 
\end{equation}

Another example is monotonic constraint where we have $C = \{ \bm\alpha: \alpha_1 \leq \alpha_2 \leq ... \leq \alpha_d\}$ and $M = \{ \bm\alpha: \alpha_1 = \alpha_2 = ... = \alpha_d\}$. The hypothesis we want to test is 
\begin{equation}
H_0: \alpha_1^* = \alpha_2^* = ... = \alpha_d^* \quad\text{versus}\quad H_A: \exists j\in \{1,...,d-1\} \,\,\, \text{s.t.} \, \,\, \alpha_j^* < \alpha_{j+1}^*.
\end{equation}



Suppose we have $n$ independent trials where we allow for $n < p$. Denote the sample negative log likelihood function as
\begin{equation}
\ell(\bm \beta) = - \frac{1}{n} \sum_{i=1}^n \log{\mathcal L_i(\bm \beta)},
\label{negative_likelihood}
\end{equation}
where $\mathcal L_i(\bm \beta)$ is the likelihood function for one trial $i$. In low dimensions we can estimate the parameter $\bm \beta$ by maximum likelihood estimation (MLE). However in high dimensions, MLE may not work. Instead we use the penalized estimator
\begin{equation}
\label{MLE}
\hat{\bm \beta} = \arg\min_{\bm \beta} \Big\{ \ell(\bm \beta) + P_{\lambda}(\bm \beta) \Big\},
\end{equation}
where $P_{\lambda}$ is some penalty function with tuning parameter $\lambda$. Note that this estimation can be performed with or without the cone constraint $\bm\alpha \in C$. In Section \ref{sec:main_theorem} we will see that all we need is the consistency of this estimator. 

Let $\nabla \ell(\bm \beta) = \nabla \ell(\bm{\alpha} , \bm \theta)$ be the gradient of the negative log likelihood function and $\nabla \ell_{\bm{\alpha}}(\bm{\alpha} , \bm \theta)$, $\nabla \ell_{\bm \theta}(\bm{\alpha} , \bm \theta)$ be the corresponding partitions. Similarly let $\nabla^2\ell(\bm \beta)$ be the sample Hessian matrix, and let $\nabla^2_{\bm{\alpha}\bm{\alpha}}\ell(\bm \beta)$, $\nabla^2_{\bm{\alpha}\bm \theta}\ell(\bm \beta)$, $\nabla^2_{\bm{\theta}\bm \alpha}\ell(\bm \beta)$ and $\nabla^2_{\bm \theta\bm \theta}\ell(\bm \beta)$ be the corresponding partitions. 
Let $H(\bm \beta)=\mathbb{E} (\nabla^2\ell(\bm \beta))$ be the population Fisher information matrix. 
Denote $H^* = H(\bm \beta^*)$ and $H^*_{\bm{\alpha}\bm{\alpha}}$, $H^*_{\bm{\alpha}\bm \theta}$, $H^*_{\bm \theta\bm \alpha}$, $H^*_{\bm \theta\bm \theta}$ as the corresponding partitions for $H^*$.


\vspace{1mm}
The difficulty of the problem comes from two aspects: the problem is high dimensional, and that we have the constraint on $\bm{\alpha}$. 
We first deal with the difficulty from high dimensions. It is well known that in low dimensions we can test for $H_0: \bm{\alpha}^* = 0$ based on the partial score function
\begin{equation}
S(\bm{\alpha}) = \nabla_{\bm{\alpha}}(\bm{\alpha}, \hat{\bm \theta}(\bm{\alpha})),
\end{equation}
where $\hat{\bm \theta}(\bm{\alpha}) = \text{argmin}_{\bm \theta}\ell(\bm{\alpha},\bm \theta)$ is the partial maximum likelihood estimator. Under the null hypothesis we have 
\begin{equation}
\sqrt{n}S(\bm 0) \overset{d}{\to} N(\bm 0, H^*_{\bm{\alpha}|\bm{\theta}}),
\end{equation}
where $H^*_{\bm{\alpha}|\bm{\theta}} = H^*_{\bm{\alpha}\bm{\alpha}} - H^*_{\bm{\alpha}\bm \theta} {H^*_{\bm \theta\bm \theta}}^{-1} H^*_{\bm \theta\bm \alpha}$ is the partial information matrix. We then reject the null when $S(\bm 0)$ is relatively large. However, in high dimensions this method does not work. To overcome this issue, we follow the decorrelation procedure introduced in \cite{fang2017testing, ning2017general} as described in Step 1 in Algorithm \ref{algo:decorrelation_constrained}.

\begin{algorithm}
   \caption{Two-step procedure for statistical inference with cone constraint}
   \label{algo:decorrelation_constrained}
   
   \vspace{3mm}
   {\bf{Step 1}}
   
\begin{enumerate}
\item[1.1] Get penalized estimator $\hat{\bm \beta} = (\hat{\bm \alpha} , \hat{\bm \theta})$ using \eqref{MLE} for some tuning parameter $\lambda$.
\item[1.2]
For each $j=1, ..., d$, estimate $\bm{\hat w}_j$ by the following Dantzig selector
\begin{equation}
\label{step1:dantzig}
\hat{\bm w}_j = \text{argmin}_{\bm w} \|\bm w\|_1 
\, \text{  s.t.  } \, \big\| \nabla^2_{\alpha_j\bm \theta}\ell(\hat{\bm \beta}) - \bm w^\top  \nabla^2_{\bm \theta\bm \theta}\ell(\widehat{\bm \beta}) \big\|_{\infty} \leq \lambda', 
\end{equation}
where $\lambda'$ is a hyper-parameter which we describe how to choose later. Combine them to get matrix $\hat{\bm{W}}$, i.e., $\hat{\bm{W}} = (\hat{\bm w}_1, ... , \hat{\bm w}_d)$.
\item[1.3] Define the decorrelated score function:
\begin{equation}
\label{step1:score}
\hat{\bm U}(\bm \alpha) = \nabla_{\bm \alpha}\ell(\bm \alpha, \hat{\bm \theta}) - \hat{\bm{W}}^\top  \nabla_{\bm \theta}\ell(\bm \alpha, \hat{\bm \theta}).
\end{equation}
\item[1.4] Define the decorrelated estimator:
\begin{equation}
\label{step1:wald}
\tilde{\bm \alpha} = \hat{\bm \alpha} - \Big(\frac{\partial\hat{\bm U}(\hat{\bm \alpha})}{\partial \bm \alpha}\Big)^{-1} \cdot \hat{\bm U}(\hat{\bm \alpha}).
\end{equation}
\item[1.5] Define the decorrelated likelihood function:
\begin{equation}
\label{step1:likeli}
\ell_{\text{de}}(\bm \alpha) = \ell\big(\bm \alpha,\bm{\hat \theta} - \bm{\hat W}(\bm \alpha - \hat{\bm \alpha}) \big).
\end{equation}
\end{enumerate}
   
   {{\bf Step 2}}
   
\begin{enumerate}
\item[2.1] Get one-sided Wald test statistic
\begin{equation}
T_w = \inf_{\bm b \in M}\Big\{ (\tilde{\bm \alpha}-\bm b)^\top  \hat H_{\bm \alpha|\bm \theta} (\tilde{\bm \alpha}-\bm b)\Big\} - \inf_{\bm b \in C}\Big\{ (\tilde{\bm \alpha}-\bm b)^\top  \hat H_{\bm \alpha|\bm \theta} (\tilde{\bm \alpha}-\bm b)\Big\}.
\end{equation}
\item[2.2] Get one-sided Likelihood ratio test statistic
\begin{equation}
T_L = 2n\Big(\inf_{\bm b \in M} \ell_{\text{de}}(\bm b) - \inf_{\bm b \in C} \ell_{\text{de}}(\bm b)\Big). 
\end{equation}
\item[2.3] Get one-sided Score test statistic
\begin{equation}
T_s = \Big( \hat{\bm U}(\bm b_M) - \hat{\bm U}(\bm b_C) \Big)^\top  \hat H_{\bm \alpha | \bm \theta}^{-1} \Big( \hat{\bm U}(\bm b_M) - \hat{\bm U}(\bm b_C) \Big), 
\end{equation}
where 
\begin{equation}
\bm b_M = \arg\inf\limits_{\bm b \in M} \ell_{\text{de}}(\bm b), \,\,\, \text{and} \,\,\, \bm b_C = \arg\inf\limits_{\bm b \in C} \ell_{\text{de}}(\bm b).
\end{equation}

\end{enumerate}
\end{algorithm}


\begin{remark}
In Step 1.2, we want to get a linear combination of $\nabla_{\bm \theta}\ell(\bm{\alpha^*}, \bm \theta^*)$ to best approximate $\nabla_{\alpha}\ell(\bm{\alpha^*}, \bm \theta^*)$. The population version of this vector should be 
\begin{equation}
\begin{aligned}
\bm W^* = \mathop{\text{argmin}}\limits_{\bm W} \, \mathbb{E} \Big\{  \nabla_{\alpha}\ell(\bm{\alpha^*}, \bm \theta^*) - {\bm W}^\top  \nabla_{\bm \theta}\ell(\bm{\alpha^*}, \bm \theta^*)  \Big\}^2 = {H^{*-1}_{\bm \theta\bm \theta}}{H^*_{\bm\theta \bm\alpha}}.
\end{aligned}
\end{equation}
However, in high dimensions, we cannot directly estimate $\bm W^*$ by the corresponding sample version since the problem is ill-conditioned. So instead we estimate $\bm W^*$ by the Dantzig selector $\hat{\bm W}$. 


\end{remark}

\begin{remark}
In Step 1.3 we get decorrelated score function which is 
approximately orthogonal to any component of the nuisance score function $\nabla_{\bm \theta} \ell(\bm{\alpha^*},\bm \theta^*)$.
This is approximately an unbiased estimating equation for $\bm \alpha$ so the root of this equation should give us an approximately unbiased estimator for $\bm \alpha^*$. Since searching for the root may be computational intensive, we use one Newton step, as stated in \eqref{step1:wald}.
\end{remark}

With the decorrelated score function, the decorrelated estimator, and the decorrelated likelihood function, under mild conditions we will specify in Section \ref{sec:main_theorem}, 
we have the following asymptotic distributions \cite{ning2017general}:
\begin{gather}
\sqrt{n}\hat{\bm U}(\bm \alpha^*) {\to} N(\bm 0, H^*_{\bm \alpha | \bm \theta}),  \\
\sqrt{n} (\tilde{\bm \alpha} - \bm \alpha^*) {\to} N(\bm 0, H_{\bm \alpha | \bm \theta}^{*-1}),  \\
2n\Big( \ell_{\text{de}}(\bm{\alpha^*}) - \ell_{\text{de}}(\bm{\tilde\alpha})  \Big) {\to} \chi^2_d, 
\end{gather}
where $H_{\bm \alpha | \bm \theta}^* = H_{\bm \alpha\bm \alpha}^* - H_{\bm \alpha\bm \theta}^*{H_{\bm \theta\bm \theta}^{*-1}}H_{\bm \theta\bm \alpha}^*$, and in practice it can be estimated by 
\begin{equation}
\hat H_{\bm\alpha| \bm \theta} = \nabla^2_{\bm\alpha\bm\alpha}\ell(\hat{\bm \beta}) - \hat{\bm W}^\top  \nabla^2_{\bm \theta\bm\alpha}\ell(\hat{\bm \beta}).
\end{equation}

We then deal with the second difficulty: cone constraint. Since we already get asymptotic normality, we follow the procedure in \cite{shapiro1988towards} to construct the Score, Wald and likelihood ratio test statistics, as described in Step 2 in Algorithm \ref{algo:decorrelation_constrained}.

This two-step procedure gives us the final test statistics $T_s$, $T_w$ and $T_L$. In the next section we will show that under null hypothesis, all of them converge weakly to the weighted Chi-square distribution, and from which we can construct valid hypothesis test with asymptotic designed Type I error.

\section{Theoretical result}
\label{sec:main_theorem}

In this section, we outline the main theoretical properties of our method. We start by providing
high-level conditions in Section \ref{sec:assumption}, and state our main theorem in Section \ref{sec:theoretical} that the null distribution is a weighted Chi-square distribution. In Section \ref{sec:weights} we describe the way to calculate the weights. We analyze the power of our method in Section \ref{sec:power} and the proof of the main theorem is given in Section \ref{sec:proof}.

\subsection{Assumptions}
\label{sec:assumption}

In this section we provide high-level assumptions that allow us to establish properties of each step in our procedure.

\paragraph{Sparsity Condition:} Both $\bm \beta^*$ and $ \bm w^*$ are sparse: $\| \bm \beta^* \|_0 = \| \bm w^* \|_0 = s$. (We use a single $s$ for notational simplicity, but this is not required for our method to work).

\paragraph{Score Condition:} The expected value of the score function at true $\bm{\beta^*}$ is 0: 
\begin{equation}
\mathbb{E}\Big( \nabla \ell(\bm \beta^*)\Big) = 0.
\end{equation}

\paragraph{Sparse Eigenvalue Condition:} We have
$\bm v^\top H^*\bm v \geq c_{\min}\|\bm v\|_2^2$ and $\bm v^\top \nabla^2\ell(\hat{\bm \beta})\bm v \geq c_{\min}\|\bm v\|_2^2$
for any $\bm v$ with $\|\bm v\|_0 = \mathcal O(s)$. 
Also both $\nabla\ell(\bm\beta^*)$, $\nabla^2\ell(\bm\beta^*)$, and $H^*$ are bounded element-wise, i.e., the maximum element is $\mathcal{O}(1)$ and each element has absolute value bounded by some constant $a$.

Denote $\| A \|_{\infty}$ as the maximum absolute value of elements in $A$, i.e., $\| A \|_{\infty} = \max_{j,k}{|A_{jk}|}$. By saying the maximum element of $H^*$ is $\mathcal{O}(1)$, we are assuming $\|H^*\|_{\infty} = \mathcal O(1)$.


\paragraph{Estimation Accuracy Condition:} The penalized estimator $\bm{\hat \beta}$ in \eqref{MLE} is a consistent estimator for the true $\bm{\beta^*}$:
\begin{equation}
\| \hat{\bm \beta} - \bm \beta^* \|_1 = \mathcal{O}(\lambda s) \,\,\, \text{and} \,\,\, \| \hat{\bm \beta} - \bm \beta^* \|_2 = \mathcal{O}(\lambda \sqrt s),
\end{equation}
where $\lambda$ is the hyper-parameter in the penalty $P_{\lambda}$.

\paragraph{Smooth Hessian Condition:} The Hessian matrix $\nabla^2 \ell(\bm \beta)$ is Lipschitz continuous:
\begin{equation}
\|\nabla^2 \ell(\bm \beta_1)-\nabla^2 \ell(\bm \beta_2)\|_{\infty} \leq L \cdot \|\bm \beta_1-\bm \beta_2\|_1, 
\end{equation}
for some constant $L$.

\begin{remark}
The score condition holds for most of the log likelihood functions. In fact, let $f$ be the likelihood function and $\theta$ be the parameter, then under certain regularity conditions \cite{inference}, we have
\begin{equation}
\begin{aligned}
\mathbb{E} \frac{d}{d\theta} \log{f} &= \mathbb{E} \frac{df}{d\theta}\cdot \frac{1}{f} =  \int \frac{df}{d\theta}\cdot \frac{1}{f} \cdot f\,dx = \int \frac{df}{d\theta} \,dx = \frac{d}{d\theta} \int f \, dx = \frac{d}{d\theta} 1 =0.
\end{aligned}
\end{equation}
\end{remark}

\begin{remark}
\label{SE_RE}
The sparse eigenvalue (SE) condition can be replaced by restricted eigenvalue (RE) condition: let $\mathcal S = supp(\bm{\beta^*}) \cup supp(\bm{w^*})$, RE condition requires 
$\bm v^\top H^*\bm v \geq c_{\min}\|\bm v\|_2^2$ and $\bm v^\top \nabla^2\ell(\hat{\bm \beta})\bm v \geq c_{\min}\|\bm v\|_2^2$
for any $\bm v$ in the cone $\mathcal C(\mathcal S) = \{ \bm v: \|\bm v_{\mathcal S^c}\| \leq c_0 \|\bm v_{\mathcal S}\|\}$ for some $c_{\min}, c_0 > 0$. Both sparse eigenvalue condition and restricted eigenvalue condition are common in high dimensional statistical estimation literature, and are known to hold for a large number of models. See Remark \ref{RE_OK} in the supplementary material for the proof.
\end{remark}

\begin{remark}
The estimation condition is also common for penalized estimators. For example, \cite{negahban2009unified} shows that, if the sample loss function $\mathcal L$ (e.g. negative log likelihood function $\ell(\bm \beta)$ here) is convex, differentiable, and satisfies Restricted Strong Convexity:
\begin{equation}
\mathcal L (\bm \beta^* + \Delta) - \mathcal L(\bm \beta^*) - \langle \nabla \mathcal L(\bm \beta^*), \Delta \rangle \geq \kappa\|\Delta\|^2
\end{equation}
for certain $\Delta$, then for $P_{\lambda}$ being $L_1$ penalty, with $\lambda \geq 2\| \nabla \mathcal L(\bm \beta^*) \|_{\infty}$ we have
\begin{equation}
\| \hat{\bm \beta} - \bm \beta^* \|_1 = \mathcal{O}(\lambda s) \,\,\, \text{and} \,\,\, \| \hat{\bm \beta} - \bm \beta^* \|_2 = \mathcal{O}(\lambda \sqrt s).
\end{equation}
\end{remark}

\begin{remark}
The smooth Hessian condition is to make sure the Hessian matrix is well-behaved locally, hence to make sure the Dantzig selector $\hat{\bm w}$ is consistent. This condition is also known to hold for general models.
\end{remark}
%

\subsection{Main theorem}
\label{sec:theoretical}

Before we proceed with our main theorem, we first introduce the following Lemma \ref{lemma_normal} which shows the asymptotic distribution of the decorrelated score function and decorrelated estimator constructed in Step 1 of Algorithm \ref{algo:decorrelation_constrained}. 
It is in the same spirit as and corresponds to Theorem 4.4 and 4.7 in \cite{fang2017testing}. 
All the other related lemmas and proofs are provided in the supplementary material.
For ease of presentation, in the following Lemma \ref{lemma_normal} we focus on the case where $\alpha$ is a scalar. 
It is straightforward to generalize to the vector case.

\begin{lemma}
\label{lemma_normal}
Suppose all the conditions in Section \ref{sec:assumption} are satisfied. Let $\lambda = \mathcal{O}( \sqrt{\log{p}/{n}})$ in Step 1.1, $\lambda' = \mathcal{O} ( s^2\sqrt{\log{p}/{n}} )$ in Step 1.2, and $s^6\log^2p/n = o(1)$, we have
\begin{gather}
\sqrt{n}\hat{ U}( \alpha^*) {\to} N( 0, H^*_{ \alpha | \bm \theta}),  \label{lemma:U_alpha}\\
\sqrt{n} (\tilde{ \alpha} -  \alpha^*) {\to} N( 0, H_{ \alpha | \bm \theta}^{*-1}), \label{lemma:alpha}  \\
\left| H^*_{\alpha | \bm \theta} - \hat H_{\alpha | \bm \theta} \right| = o_{\PP}(1) \label{lemma:H_consistent}
\end{gather}
where $H_{ \alpha | \bm \theta}^* = H_{ \alpha \alpha}^* - H_{ \alpha\bm \theta}^*{H_{\bm \theta\bm \theta}^{*-1}}H_{\bm \theta \alpha}^*$ and is estimated by the sample version
\begin{equation}
\label{eq:H_sample}
\hat H_{\alpha| \bm \theta} = \nabla^2_{\alpha\alpha}\ell(\hat{\bm \beta}) - \hat{\bm w}^\top  \nabla^2_{\bm \theta\alpha}\ell(\hat{\bm \beta}).
\end{equation}
\end{lemma}

\begin{proof}
The outline of the proof follows from \cite{fang2017testing}. 
We start from the proof of \eqref{lemma:U_alpha} for $\hat U(\alpha^*)$ where by mean value theorem we have:
\begin{equation}
\begin{aligned}
\hat U(\alpha^*) &=  \nabla_{\alpha}\ell(\alpha^*, \hat{\bm  \theta}) - \hat{\bm  w}^\top  \nabla_{\bm \theta}\ell(\alpha^*, \hat{\bm  \theta}) \\
&= \nabla_{\alpha}\ell(\alpha^* , {\bm  \theta}^*) + \nabla^2_{\alpha \bm \theta} \ell(\alpha^* , \bar{\bm \theta})(\hat{\bm \theta} - {\bm \theta}^*) - [ \hat{\bm  w}^\top  \nabla_{\bm \theta}\ell(\alpha^* , {\bm  \theta}^*) + \hat{\bm  w}^\top \nabla^2_{\bm \theta \bm \theta} \ell(\alpha^* , \tilde{\bm \theta})(\hat{\bm \theta} - {\bm \theta}^*)  ] \\
&= \underbrace{\big[\nabla_{\alpha}\ell(\alpha^* , {\bm  \theta}^*) - {\bm  w}^{*T} \nabla_{\bm \theta}\ell(\alpha^* , {\bm  \theta}^*)\big]  }_{E_1}
+ \underbrace{\big[({\bm  w}^{*} - \hat{\bm w})^\top  \nabla_{\bm \theta}\ell(\alpha^* , {\bm  \theta}^*)\big]}_{E_2}\\
&\qquad\qquad\qquad\qquad\qquad\qquad + \underbrace{\big[\nabla^2_{\alpha \bm \theta} \ell(\alpha^* , \tilde{\bm \theta}) - \hat{\bm  w}^\top \nabla^2_{\bm \theta \bm \theta} \ell(\alpha^* , \bar{\bm \theta})\big](\hat{\bm \theta} - {\bm \theta}^*)}_{E_3} \\
&= E_1 + E_2 + E_3,
\end{aligned}
\end{equation}
where $\bar{\bm \theta} = {\bm \theta}^* + \bar u(\hat{\bm \theta} - {\bm \theta}^*), \tilde{\bm \theta} = {\bm \theta}^* + \tilde u(\hat{\bm \theta} - {\bm \theta}^*)$ for some $\bar u, \tilde u \in [0,1]$.

We consider the three terms separately. For $E_1$, by taking $\bm v = (1; -{\bm w}^*)$ in Lemma \ref{lemma:CLT}, under the null hypothesis we have
\begin{equation}
\sqrt n E_1 {\to}  N(0, H_{\alpha | \bm \theta}^*).
\end{equation}
For $E_2$, according to H\"{o}lder's inequality, Lemma \ref{lemma:grad_true}, and Lemma \ref{lemma:diff_w} we have
\begin{equation}
\begin{aligned}
|E_2| \leq \|\widehat{\bm w} - \bm w^*\|_1 \cdot \|\nabla_{\bm \theta}\ell(0, {\bm  \theta}^*)\|_{\infty} = \mathcal O_{\mathbb P}(\lambda's\sqrt{\log p/n}) = \mathcal{O}_\PP({s }^{3} \log p/n).
\end{aligned}
\end{equation}
For $E_3$ we have
\begin{equation}
\begin{aligned}
|E_3| &\leq \underbrace{|  (\hat {\bm w} - {\bm w}^*) ^\top \nabla^2_{{\bm \theta}{\bm \theta}} \ell(\alpha^*, \bar{\bm \theta})(\hat{\bm \theta} - {\bm \theta}^*) |}_{R_1} + 
\big| [\nabla^2_{\alpha{\bm \theta}} \ell(\alpha^*, \tilde{\bm \theta}) - {{\bm w}^*}^\top \nabla^2_{{\bm \theta}{\bm \theta}} \ell(\alpha^*, \bar{\bm \theta})](\hat{\bm \theta} - {\bm \theta}^*) \big| \\
&\leq R_1 + \underbrace{\big| [\nabla^2_{\alpha{\bm \theta}} \ell(\alpha^*, \tilde{\bm \theta}) - H^*_{\alpha{\bm \theta}}](\hat{\bm \theta} - {\bm \theta}^*) \big|}_{R_2} + 
\big| [H^*_{\alpha{\bm \theta}} - {{\bm w}^*}^\top \nabla^2_{{\bm \theta}{\bm \theta}} \ell(\alpha^*, \bar{\bm \theta})](\hat{\bm \theta} - {\bm \theta}^*) \big| \\
&\leq R_1 + R_2 + \underbrace{ \big| {{\bm w}^*}^\top [H^*_{{\bm \theta}{\bm \theta}} -  \nabla^2_{{\bm \theta}{\bm \theta}} \ell(\alpha^*, \bar{\bm \theta})](\hat{\bm \theta} - {\bm \theta}^*) \big| }_{R_3}\\
&\leq R_1 + R_2 + R_3.
\end{aligned}
\end{equation}

Considering the three terms $R_1, R_2$ and $R_3$ separately, according to Lemma \ref{lemma:H} and Lemma \ref{lemma:diff_w} we have
\begin{gather}
R_1 \leq \|\hat {\bm w} - {\bm w}^*\|_2 \| \nabla_{{\bm \theta}{\bm \theta}} \ell (\alpha^*, \bar{\bm \theta})\|_2 \|\hat{\bm \theta} - {\bm \theta}^*\|_2 \leq C\sqrt{s } \lambda'  c  \sqrt{s \log p/n} = \mathcal{O}_\PP({s }^3 \log p/n),  \\
R_2 \leq \| \nabla^2_{\alpha{\bm \theta}} \ell(\alpha^*, \tilde{\bm \theta}) - H^*_{\alpha{\bm \theta}} \|_{\infty} \cdot \|\hat{\bm \theta} - {\bm \theta}^*\|_1   = \mathcal{O}_\PP({s}^2 \log p/n),  \\
R_3 \leq \|{{\bm w}^*}\|_1 \| H^*_{{\bm \theta}{\bm \theta}} -  \nabla^2_{{\bm \theta}{\bm \theta}} \ell(\alpha^*, \bar{\bm \theta}) \|_{\infty} \|\hat{\bm \theta} - {\bm \theta}^*\|_1 = \mathcal{O}_\PP({s }^{3} \log p/n). 
\end{gather}

Combining all these terms we show that \eqref{lemma:U_alpha} holds. 
We then turn to the proof of \eqref{lemma:H_consistent} that $\hat H_{\alpha | \bm \theta}$ is an consistent estimator. By definition we have
\begin{equation}
\begin{aligned}
| \hat H_{\alpha | \bm \theta} -  H_{\alpha | \bm \theta}^* | &\leq \underbrace{ | H^*_{\alpha\alpha} - \nabla^2_{\alpha\alpha} \ell(\hat \alpha, \hat{\bm \theta}) | }_{T_1} + | H_{ \alpha\bm \theta}^*{H_{\bm \theta\bm \theta}^{*-1}}H_{\bm \theta \alpha}^* - \hat{\bm  w}^\top \nabla^2_{\bm \theta  \alpha} \ell(\hat\alpha , \hat{\bm \theta}) | \\
& \leq T_1 + \underbrace{|(\bm w^* - \bm{\hat w})^\top H_{\bm \theta \alpha}^* |}_{T_2} 
+ \underbrace{ | \hat{\bm  w}^\top ( H_{\bm \theta \alpha}^* - \nabla^2_{\bm \theta  \alpha} \ell(\hat\alpha , \hat{\bm \theta})) | }_{T_3} \\
& \leq T_1 + T_2 + T_3.
\end{aligned}
\end{equation}

Considering the terms $T_1, T_2$ and $T_3$ separately, according to Lemma \ref{lemma:H} and Lemma \ref{lemma:diff_w} we have
\begin{gather}
T_1 = \mathcal{O}_\PP({s } \sqrt{\log p/n}),  \\
T_2 \leq \|\bm w^* - \bm{\hat w}\|_{1} \cdot \|H_{\bm \theta \alpha }^*\|_{\infty}  = \mathcal{O}_\PP({s}^3 \sqrt{\log p/n}),  \\
T_3 \leq \| \hat{\bm  w} \|_{1} \cdot \| H_{\bm \theta \alpha }^* - \nabla^2_{\bm \theta  \alpha } \ell(\hat\alpha , {\hat \theta}) \|_{\infty} = \mathcal{O}_\PP({s}^2 \sqrt{\log p/n}). 
\end{gather}

Combining the three terms we show that \eqref{lemma:H_consistent} holds. Finally we prove the result \eqref{lemma:alpha} for $\tilde{ \alpha }$. By construction we have
\begin{equation}
\begin{aligned}
\label{eq:alpha_tilde_bound}
\tilde{ \alpha } &= \hat{ \alpha } - \Big(\frac{\partial\hat{U}(\hat{ \alpha })}{\partial  \alpha }\Big)^{-1} \cdot \hat{U}(\hat{ \alpha }) 
= \hat{ \alpha } - H^{*-1}_{\alpha | \bm \theta} \hat{U}(\hat{ \alpha }) + 
\underbrace{\hat{U}(\hat{ \alpha }) \Big[ H^{*-1}_{\alpha | \bm \theta} - \Big(\frac{\partial\hat{U}(\hat{ \alpha })}{\partial  \alpha }\Big)^{-1} \Big]}_{S_1} \\
&= \hat{ \alpha } - H^{*-1}_{\alpha | \bm \theta} \Big[ \hat{U}({ \alpha ^*}) + (\hat{ \alpha } - { \alpha ^*}) \cdot \frac{\partial\hat{U}(\breve{ \alpha })}{\partial  \alpha }  \Big] + S_1 \\
&= \hat{ \alpha } - H^{*-1}_{\alpha | \bm \theta} \hat{U}({ \alpha ^*}) - (\hat{ \alpha } - { \alpha ^*}) H^{*-1}_{\alpha | \bm \theta} \cdot H^{*}_{\alpha | \bm \theta} 
+ \underbrace{ (\hat{ \alpha } - { \alpha ^*}) H^{*-1}_{\alpha | \bm \theta} \Big[ H^{*}_{\alpha | \bm \theta} - \Big(\frac{\partial\hat{U}(\breve{ \alpha })}{\partial  \alpha }\Big) \Big] }_{S_2} + S_1 \\
&= { \alpha ^*} - H^{*-1}_{\alpha | \bm \theta} \hat{U}({ \alpha ^*}) + S_1 + S_2,
\end{aligned}
\end{equation}
where $\breve{ \alpha } = { \alpha ^*} + \breve u (\hat{ \alpha }  - { \alpha ^*})$ for some $\breve u \in [0,1]$.
We consider the terms $S_1$ and $S_2$ separately. For $S_1$ we have
\begin{equation}
| \hat{U}({ \alpha ^*}) - \hat{U}(\hat{ \alpha }) | \leq | { \alpha ^*} - \hat{ \alpha } | \cdot \Big|\frac{\partial\hat{U}(\breve{ \alpha })}{\partial  \alpha } \Big| = \mathcal{O}_\PP( \lambda ).
\end{equation}

Moreover, from the analysis of $\hat{U}({ \alpha ^*})$ above we have that $|\hat{U}({ \alpha ^*})| = \mathcal{O}_\PP(n^{-1/2}) $. We then obtain
\begin{equation}
\label{eq:S_1}
|S_1| \leq \Big( | \hat{U}({ \alpha ^*}) - \hat{U}(\hat{ \alpha }) | + |\hat{U}({ \alpha ^*})| \Big) \cdot 
\Big[ H^{*-1}_{\alpha | \bm \theta} - \Big(\frac{\partial\hat{U}(\hat{ \alpha })}{\partial  \alpha }\Big)^{-1} \Big]
\leq \mathcal{O}_\PP({s}^3 {\log p/n}).
\end{equation}

For $S_2$ we have that
\begin{equation}
\label{eq:S_2}
|S_2| \leq | \hat{ \alpha } - { \alpha ^*} | \cdot H^{*-1}_{\alpha | \bm \theta} \cdot \Big| H^{*}_{\alpha | \bm \theta} - \Big(\frac{\partial\hat{U}(\breve{ \alpha })}{\partial  \alpha }\Big) \Big| \leq \mathcal{O}_\PP({s}^3 {\log p/n}).
\end{equation}

Plugging in \eqref{eq:S_1} and \eqref{eq:S_2} into \eqref{eq:alpha_tilde_bound} we obtain
\begin{equation}
\sqrt{n} (\tilde{ \alpha } -  \alpha ^*) = - \sqrt{n} H^{*-1}_{\alpha | \bm \theta} \hat{U}({ \alpha ^*}) + o_\PP(1).
\end{equation}

According to \eqref{lemma:U_alpha}, this gives
\begin{equation}
\sqrt{n} (\tilde{ \alpha } -  \alpha ^*) {\to} N( 0, H_{ \alpha  | \bm \theta}^{*-1}),
\end{equation}
and our claim \eqref{lemma:alpha} holds.
\end{proof}

%
%

\begin{remark}
The stated sample complexity $s^6\log^2p/n = o(1)$ is for a general model. For specific models we may be able to get sharper results. 
For example for linear model and generalized linear model $s^2\log^2p/n = o(1)$ suffices \cite{ning2017general}. 
\end{remark}

%

\vspace{3mm}
In Lemma \ref{lemma_normal} we focus on the case where $\alpha$ is a scalar. 
It is straightforward to generalize to the vector case.
We are now almost ready for our main theorem. 
For any positive definite matrix $V$, denote $\langle x,y\rangle_V = x^\top Vy$ and $\|x\|_V = (x^\top Vx)^{\frac 12}$ as the inner product and the norm, respectively. 
For the linear space $M$, the usual orthogonal complement of $M$ associated with $V$ is defined as
\begin{equation}
M^\perp_V = \Big\{y: \langle x,y \rangle_V = 0 \,\,\, \text{for all} \,\,\, x \in M \Big\}. 
\end{equation}

For any positive definite matrix $V\in \RR^{m \times m}$ and convex cone $C\subseteq \RR^m$, let $y \sim N(0,V)$ and consider 
\begin{equation}
\label{eq:def_bar_chi}
T_0 = y^\top V^{-1}y - \min_{\eta \in C} (y-\eta)^\top V^{-1}(y-\eta).
\end{equation}

It can be shown \cite{shapiro1988towards} that $T_0$ is distributed as a weighted mixture of Chi-squared distribution associated with $V $ and $C$ denoted as $T_0 \sim \bar \chi^2(V, C)$. That is 
\begin{equation}
\label{eq:tail_c_prob}
\Pr\big\{ T_0 \geq c \big\} = \Pr\big\{ \bar \chi^2(V, C) \geq c \big\} = \sum_{i=0}^m w_i(m, V, C) \cdot \Pr\big\{ \chi^2_i \geq c \big\}, 
\end{equation}
where $\chi_i^2$ is a Chi-squared random variable with $i$ degrees of freedom and $\chi_0^2$ is the point mass at 0. Here $w_i(m, V, C)$ are non-negative weights satisfying $\sum_{i=1}^m w_i(m, V, C) = 1$. See Section \ref{sec:weights} for details. 
We then have the following main theorem:
\begin{theorem}
Suppose the hypothesis we would like to test is $H_0 : \alpha^* \in M$ versus $H_A : \alpha^* \in C\backslash M$ where we have the constraint $\alpha^* \in C$, and suppose all the conditions in Section \ref{sec:assumption} are satisfied. Then under the null hypothesis, the test statistics $T_s$, $T_w$ and $T_L$ constructed in Step 2 satisfy
\begin{equation}
\label{eq:null_distribution}
T_s, T_w, T_L {\to} \bar \chi^2(H^*_{\bm \alpha | \bm \theta}, C^*), 
\end{equation}
where $C^* = C \cap M^\perp_{H^*_{\bm \alpha | \bm \theta}}$.
\label{theorem_main}
\end{theorem}
The proof of Theorem \ref{theorem_main} is postponed to Section \ref{sec:proof}. 

\begin{remark}
Our method is also valid for cones not centered at the origin, for example $C = \{\bm \alpha: R\bm \alpha \geq r \}$ and $M = \{\bm \alpha: R\bm \alpha = r \}$. The two-step procedure is exactly the same as before. Under the null hypothesis $R\bm \alpha^*=r$, by removing $\bm \alpha^*$ from both $\tilde{\bm \alpha}$ and $\bm b$, we see that $T_w$ has the same distribution with the case $C = \{\bm \alpha: R\bm \alpha \geq 0 \}$ and $M = \{\bm \alpha: R\bm \alpha = 0 \}$. 
This is also validated experimentally by the sum constraint in Section \ref{sec:synthetic_data}.
\end{remark}

\begin{remark}
In this paper we focus on hypothesis on a low dimension parameter $\bm \alpha$ only. It is in fact straightforward to extend Theorem \ref{theorem_main} to the whole parameter $\bm \beta$. However, as we will see in Section \ref{sec:weights},  the weights of the null distribution \eqref{eq:null_distribution} usually lack closed form expression and can only be calculated using numerical methods in practice. 
When dimension of parameter of interest is large, this could be computationally intractable. 
\end{remark}

\vspace{1mm}
With this weighted Chi-square distribution under the null, we can build hypothesis test for $\bm \alpha^*$ with any designed Type I error. It remains to calculate the weights $w_i$ and the critical value $c$ in \eqref{eq:tail_c_prob}. We describe the calculation of the weights in the next section. The critical value can be calculated numerically as follow.

\paragraph{Critical value.} The final step is to calculate the critical value. Specifically, we want to find critical value $c$ such that 
\begin{equation}
\Pr\Big\{ \bar \chi^2 \big(\hat H_{\bm \alpha| \bm \theta}, C^*\big) \geq c \Big\} = \sum_{i=0}^m w_i\big(m, \hat H_{\bm \alpha| \bm \theta}, C^*\big) \cdot \Pr\big\{ \chi^2_i \geq c \big\} = \gamma, 
\end{equation}
where $\gamma$ is the designed Type I error. This can be solved numerically by binary search on $c$. 

Combining all these result we are able to build valid testing procedure for the original hypothesis \eqref{eq:main_hypothesis} with asymptotic designed Type I error $\gamma$ by calculating the $1-\gamma$ quantile of the weighted Chi-squared distribution, and reject $H_0$ when $T_s$, $T_w$ or $T_L$ is greater than this quantile.

\subsection{Weights Calculation}
\label{sec:weights}

According to Lemma \ref{lemma_normal}, the covariance matrix $H^*_{\bm \alpha | \bm \theta}$ can be consistently estimated by sample version \eqref{eq:H_sample}. The cone $C^*$ depends on the constraint space $C$, $M$ and $H^*_{\bm \alpha | \bm \theta}$. For example for non-negative constraint, we have $M = \{ \bm\alpha: \bm\alpha = \bm 0\}$ and hence $M^\perp_{H^*_{\bm \alpha | \bm \theta}} = \RR^d$  and $C^* = C \cap M^\perp_{H^*_{\bm \alpha | \bm \theta}} = C$; for monotonic constraint, we have $M = \{ \bm\alpha: \alpha_1 = \alpha_2 = ... = \alpha_d\}$ and hence $M^\perp_{H^*_{\bm \alpha | \bm \theta}} = \{\bm\alpha: \bm 1^\top  H^*_{\bm \alpha | \bm \theta} \cdot \bm \alpha = 0\}$. Since $C = \{ \bm\alpha: \alpha_1 \leq \alpha_2 \leq ... \leq \alpha_d\}$, we have 
\begin{equation}
C^* = C \cap M^\perp_{H^*_{\bm \alpha | \bm \theta}} = \{ \bm\alpha: \alpha_1 \leq \alpha_2 \leq ... \leq \alpha_d, \bm 1^\top  H^*_{\bm \alpha | \bm \theta} \cdot \bm \alpha = 0 \}. 
\end{equation}

The weights $w_i(d, H^*_{\bm \alpha | \bm \theta}, C^*)$ depend on $H^*_{\bm \alpha | \bm \theta}$ and $C^*$ and can be complicated and without closed form expression. Here we briefly review the expression of general weights $w_i(m,V,C)$ for some general dimension $m$, covariance matrix $V$, and cone $C$ obtained in \cite{kudo1963multivariate}. 
We refer to \cite{shapiro1988towards} for more detailed formulas. We start from the simplest case where $C = \RR^m_+$ and $V = I$. From \eqref{eq:def_bar_chi} we have
\begin{equation}
T_0 = \sum_{i=1}^m \max(y_i, 0)^2 \sim \bar \chi^2(I, \RR^m_+).
\end{equation}

We can see that the weight $w_i(m, I, \RR^m_+)$ depends on the number of positive components of $y$: if $i$ of them are positive then the distribution would be $\chi_i^2$. There are in total $2^m$ choices of signs on each component of $y$ and therefore
\begin{equation}
w_i(m, I, \RR^m_+) = 2^{-m}\binom{m}{i}.
\end{equation}

We then consider $C = \RR^m_+$ with general $V$ where the weights are given by
\begin{equation}
w_i(m,V,\RR^m_+) \coloneqq \sum_{|\mathcal A| = i, \, \mathcal A \subseteq [m]} p\Big\{ (V_{\mathcal A^c})^{-1} \Big\} \cdot p\Big\{ V_{\mathcal A;\mathcal A^c} \Big\},
\label{wi}
\end{equation}
where the summation runs over all subsets $\mathcal A$ of $\{1, ..., m\}$ having $i$ elements. $\mathcal A^c$ is the complement of $\mathcal A$ and $V_{\mathcal A}$ is the submatrix of $V$ corresponding to those $y_i$ where $i \in \mathcal A$. $V_{\mathcal A;\mathcal A^c}$ is the covariance matrix under the condition $y_j = 0$ where $j \in \mathcal A^c$. Finally $p( \Lambda)$ denotes the probability that $ z \geq 0$ for a Gaussian random variable $ z \sim N( 0,  \Lambda)$.

The weight \eqref{wi} can be approximated using Monte Carlo simulation when $m$ is relatively small. For large $m$ this could be computational intensive, but since we are interested in $\bm\alpha \in \RR^d$ with $d \ll p$ we expect $d$ to be relatively small. 

\vspace{1mm}
We then consider more general cones $C$ and show how they can be reduced to the above case $C = \RR^m_+$ as proposed in \cite{shapiro1988towards}. First suppose $C$ is defined by linear inequality constraints
\begin{equation}
C_R = \{ \bm \alpha: R\bm \alpha \geq 0 \},
\end{equation}
where $R \in \RR^{m \times m}$ is nonsingular. In this case by linear transformation we have 
\begin{equation}
w_i(m,V,C_R) = w_i \big(m, RVR^\top , \RR^m_+ \big).
\end{equation}

More generally suppose $R \in \RR^{k \times m}$ with rank $k$, we have
\begin{equation}
w_{m-k+j}(m,V,C_R) = w_j \big(k, RVR^\top , \RR^m_+ \big),
\end{equation}
and the remaining weights vanish.

\vspace{1mm}
Finally consider the standard linear constraint with $C = \{\bm \alpha: R\bm \alpha \geq 0 \}$ and $M = \{\bm \alpha: R\bm \alpha = 0 \}$ where $R \in \RR^{k \times m}$ has full row rank. In this case we can calculate the final weights associated with $C^*$ directly. We first find $A \in \RR^{(m-k) \times m}$ as the null space of $RV$ satisfying $RVA^\top  = 0$. Then the cone $C^*$ is given by
\begin{equation}
C^* = \{\bm \alpha: R\bm \alpha \geq 0, A\bm \alpha = 0\},
\end{equation}
and the final weights associated with $C^*$ are given by
\begin{equation}
w_j(m, V, C^*) = w_j(k, RVR^\top , \RR^m_+ \big).
\end{equation}


\subsection{Power analysis}
\label{sec:power}

In this section we analyze the power of our proposed method and compare with the standard method where the constraints are ignored. Since it is unclear how to define the margin and alternative hypothesis for general cone constraint, in this section we focus on the nonnegativity constraint. The idea can be generalized to general cone straightforwardly.

We start from the scalar case where $d = 1$ and $\alpha$ is a scalar. In this case we want to test for 
\begin{equation}
H_0: \alpha^* = 0 \quad\text{versus}\quad H_A: \alpha^* > 0.
\end{equation}

To ease calculation we assume we have $n = 1$ sample and the variance is known as $\sigma^2 = H_{\bm \alpha | \bm \theta}^* = 1$. Since the three tests are asymptotically equivalent, we focus on Wald test only. According to Step 2 we have $T_w = \max( \tilde \alpha, 0)^2$ where $\tilde \alpha {\to} N(\alpha^*, 1)$ by Lemma \ref{lemma_normal}. Under the null hypothesis $\alpha^* = 0$, the asymptotic null distribution of $T_w$ is given by
\begin{equation}
T_w {\to} \frac 12 \chi^2_0 + \frac 12 \chi^2_1.
\end{equation}

Based on this asymptotic null distribution we reject the null hypothesis when $T_w$ is large. For standard Wald test where the nonnegativity constraint is ignored, the asymptotic null distribution is $\tilde \alpha {\to} N(0, 1)$ and we reject the null hypothesis when $\tilde \alpha$ is large.
Denote $\Phi(\cdot)$ as the cumulative distribution function of standard normal variable. 
Given the designed Type I error $\gamma$, the critical value for standard method is given by $\tau_1 = \Phi^{-1}(1-\gamma/2)$ and the critical value for our method is given by $\tau_2 = \Phi^{-1}(1-\gamma)$. 
Under the alternative hypothesis that $\alpha^* > 0$, the power of the standard method is given by
\begin{equation}
\PP_{\text{standard}}(\text{reject} \, | \, \alpha^*) = 1 - \Phi(\tau_1 - \alpha^*) + \Phi(-\tau_1 - \alpha^*), 
\end{equation}
while the power of our method is given by
\begin{equation}
\PP_{\text{our}}(\text{reject} \, | \, \alpha^*) =1 - \Phi(\tau_2 - \alpha^*). 
\end{equation}

It is straightforward to calculate that 
\begin{equation}
f(\alpha^*) \coloneqq \PP_{\text{our}}(\text{reject} \, | \, \alpha^*) - \PP_{\text{standard}}(\text{reject} \, | \, \alpha^*) = \big[ \Phi(\tau_1 - \alpha^*) - \Phi(\tau_2 - \alpha^*) \big] - \Phi(-\tau_1 - \alpha^*).
\end{equation}

For small $\alpha^*$ ($0 < \alpha^* \leq \tau_1 + \tau_2$), we have 
\begin{equation}
\Phi(\tau_1 - \alpha^*) - \Phi(\tau_2 - \alpha^*) \geq \frac{\gamma}{2} > \Phi(-\tau_1 - \alpha^*), 
\end{equation}
and hence $f(\alpha^*) > 0$. For large $\alpha^*$ ($\alpha^* > \tau_1 + \tau_2$), we write $\alpha^* = \tau_1 + \tau_2 + \epsilon$ with some $\epsilon > 0$ and rewrite $f(\alpha^*)$ as
\begin{equation}
g(\epsilon) \coloneqq f(\alpha^*) = \Phi(-\tau_2 - \epsilon) - \Phi(-\tau_1 - \epsilon) - \Phi(-2\tau_1 - \tau_2 - \epsilon).
\end{equation}

Clearly we have $g(0) > 0$ and $g(+\infty) = 0$ and also
\begin{equation}
\begin{aligned}
g'(\epsilon) &= -\phi(-\tau_2 - \epsilon) + \phi(-\tau_1 - \epsilon) + \phi(-2\tau_1 - \tau_2 - \epsilon) \\
&= -\phi(\tau_2 + \epsilon) + \phi(\tau_1 + \epsilon) + \phi(2\tau_1 + \tau_2 + \epsilon), 
\end{aligned}
\end{equation}
where $\phi(\cdot)$ denotes the probability density function of standard normal variable. Since $\phi(\cdot)$ decays exponentially, some simple calculation shows that $g'(\epsilon) < 0$ for any $\epsilon > 0$. Together with the fact that $g(0) > 0$ and $g(+\infty) = 0$, we know $g(\epsilon) > 0$ for any $\epsilon > 0$, which indicates that $f(\alpha^*) > 0$ for $\alpha^* > \tau_1 + \tau_2$. Therefore, for any $\alpha^* > 0$, we have $f(\alpha^*) > 0$ which shows that our method has greater power than the standard method.

Figure \ref{power_analysis_d1} shows the powers obtained by our method and standard method for $\gamma = 0.05$. We can see that when $\alpha^* = 0$ (i.e. under the null) both methods have Type I error 0.05. As $\alpha^*$ increases and the null is violated, our method has much larger power compared to the standard method. Finally when $\alpha^*$ is sufficiently large, both methods has power close to 1.

We then turn to the vector case where $\bm \alpha \in \RR^d$. Again to ease calculation we assume we have $n = 1$ sample and the variance is known as $\Sigma = H_{\bm \alpha | \bm \theta}^* = I_d$.
In this case we have $\tilde{\bm \alpha} {\to} N(\bm{\alpha}^*, I_d)$ by Lemma \ref{lemma_normal}, and the asymptotic null distribution of $T_w$ is given by
\begin{equation}
T_w {\to} \sum_{i=0}^d 2^{-d}\binom{d}{i} \chi^2_i.
\end{equation}

To violate the null hypothesis, we increase the value $\alpha_1^*$ and $\alpha_2^*, ..., \alpha_d^*$ remain to be 0. 
Figure \ref{power_analysis_d4} shows the comparison result for $d=4$ and $\gamma = 0.05$. The pattern is similar to Figure \ref{power_analysis_d1}.

\begin{figure*}[htbp]
\begin{minipage}[t]{0.5\linewidth}
\centering
\includegraphics[width=0.95\textwidth]{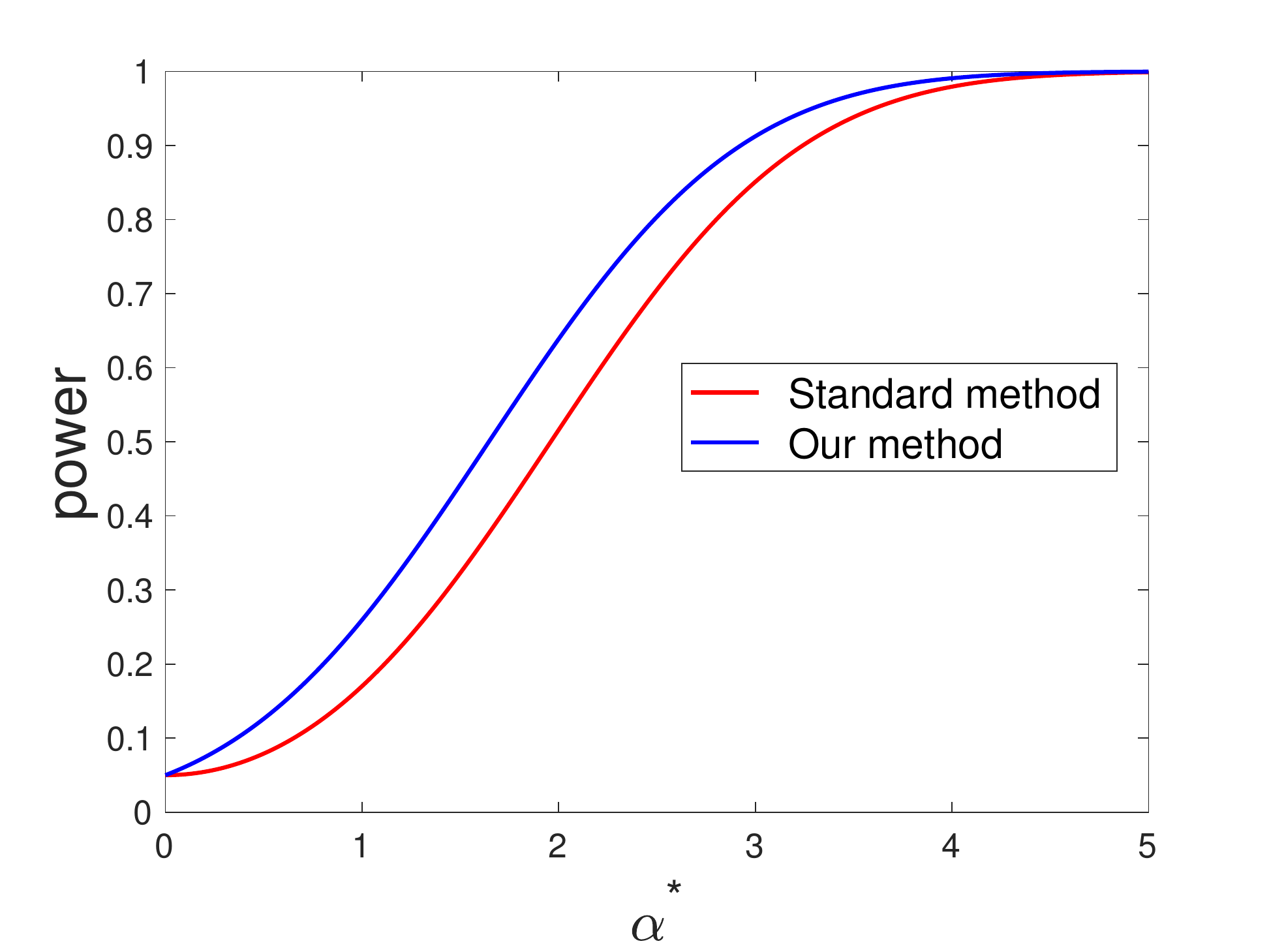}
\caption{Comparison of power with $d=1$}
\label{power_analysis_d1}
\end{minipage}
\begin{minipage}[t]{0.5\linewidth}
\centering
\includegraphics[width=0.95\textwidth]{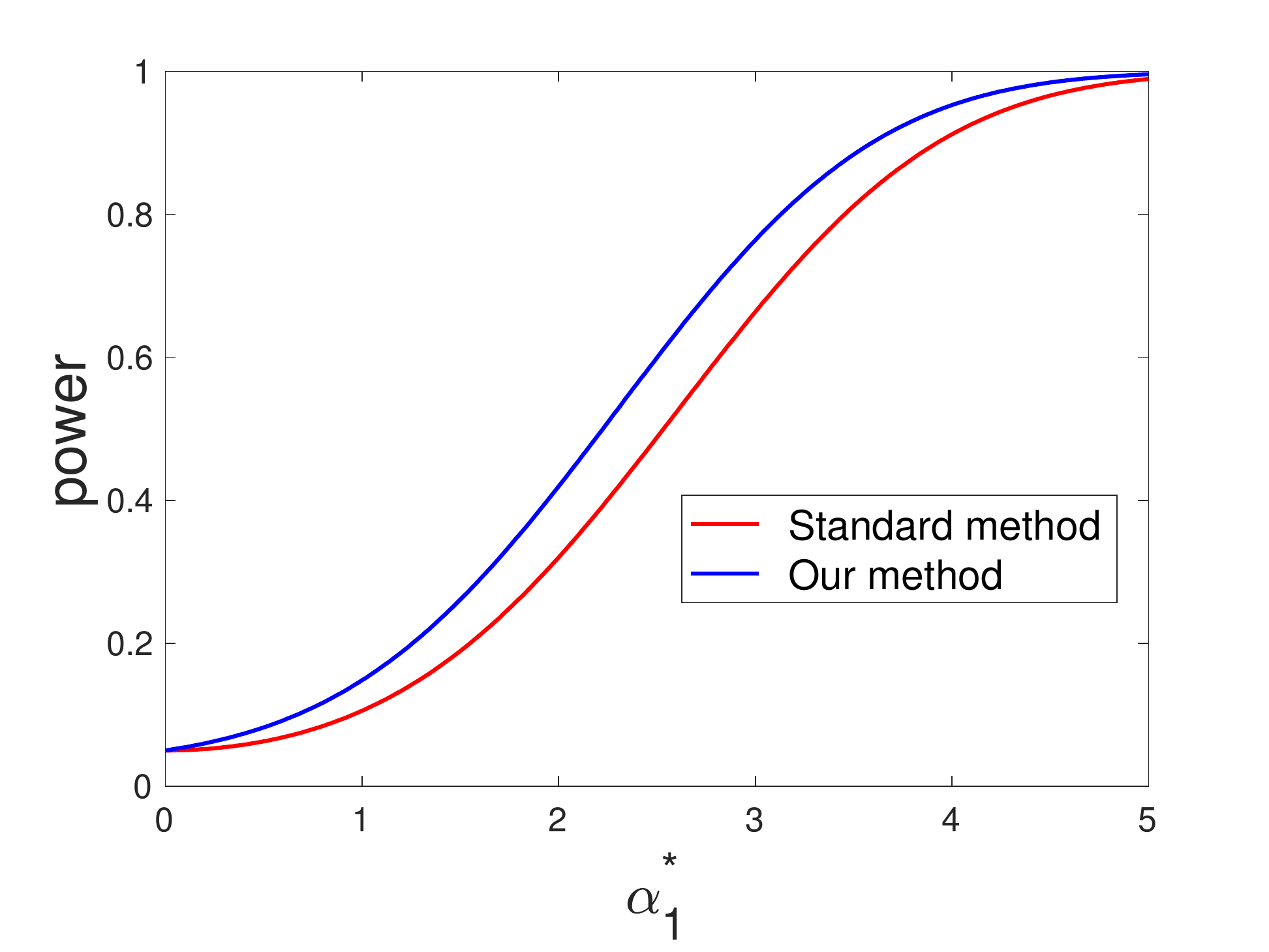}
\caption{Comparison of power with $d=4$}
\label{power_analysis_d4}
\end{minipage}
\end{figure*}

\subsection{Proof of Theorem \ref{theorem_main}}
\label{sec:proof}

Before we proceed with the main proof, we first introduce the following lemma in \cite{shapiro1988towards}.
\begin{lemma}
For any positive definite matrix $V\in \RR^{m \times m}$, convex cone $C\subseteq \RR^m$ and linear space $M \subseteq C$, let $y \sim N(\mu,V)$ with $\mu \in M$. Then the statistic
\begin{equation}
\label{eq:chi_bar_2}
T = \min_{\eta \in M} (y-\eta)^\top V^{-1}(y-\eta) - \min_{\eta \in C} (y-\eta)^\top V^{-1}(y-\eta)
\end{equation}
has the distribution $\bar \chi^2(V, C^*)$ where $C^* = C \cap M^\perp_{V^{-1}}$.
\label{lemma:shapiro2}
\end{lemma}

\begin{proof}
Denote $P(\cdot, C)$ as the orthogonal projection onto $C$ according to norm $\| \cdot \|_{V^{-1}}$. We have
\begin{equation}
\big\| y - P(y,C) \big\|_{V^{-1}}^2 = \min_{\eta \in C} (y-\eta)^\top V^{-1}(y-\eta).
\end{equation}

Since $M$ is linear space, we have the Pythagoras' theorem
\begin{equation}
\big\| y - P(y,M) \big\|_{V^{-1}}^2 = \big\| y - P(y,C) \big\|_{V^{-1}}^2 + \big\| P(y,C^*) \big\|_{V^{-1}}^2.
\end{equation}

Then \eqref{eq:chi_bar_2} follows directly from \eqref{eq:def_bar_chi}.
\end{proof}

\vspace{2mm}
We then proceed with the proof of Theorem \ref{theorem_main}. According to Lemma \ref{lemma_normal} we have
\begin{equation}
\sqrt{n} (\tilde{\bm \alpha} - {\bm \alpha}^*) {\to} N(\bm 0, H_{{\bm \alpha} | \bm \theta}^{*-1}).
\end{equation}

Under the null ${\bm \alpha}^* \in M$, Lemma \ref{lemma:shapiro2} immediately indicates that $T_w {\to} \bar \chi^2(H^*_{{\bm \alpha} | \bm \theta}, C^*)$, 
where $C^* = C \cap M^\perp_{H^*_{{\bm \alpha} | \bm \theta}}$. We then show that $T_L$ and $T_s$ are asymptotically equivalent to $T_w$. For $T_L$, following Proposition 4.2.2 in \cite{silvapulle2011constrained}, we have the local quadratic approximation
\begin{equation}
\begin{aligned}
\ell_{\text{de}}(\bm b) &= \ell_{\text{de}}(\tilde {\bm \alpha}) + \frac{\partial \ell_{\text{de}}({\bm \alpha})}{\partial{\bm \alpha}}\bigg\rvert^\top _{{\bm \alpha} = \tilde {\bm \alpha}}(\bm b-\tilde{\bm \alpha}) + \frac 12(\bm b-\tilde{\bm \alpha})^\top  \frac{\partial^2 \ell_{\text{de}}({\bm \alpha})}{\partial{\bm \alpha}^2}\bigg\rvert_{{\bm \alpha} = \tilde {\bm \alpha}} (\bm b-\tilde{\bm \alpha}) + o_p(1) \\
&= \ell_{\text{de}}(\tilde {\bm \alpha}) + \frac 12 (\bm b-\tilde{\bm \alpha})^\top  \hat H_{{\bm \alpha} | \bm \theta} (\bm b-\tilde{\bm \alpha}) + o_p(1),
\end{aligned}
\end{equation}
where the second term follows from the definition of $\bm{\tilde {\bm \alpha}}$ and the following Taylor expansion
\begin{equation}
\frac{\partial \ell_{\text{de}}({\bm \alpha})}{\partial{\bm \alpha}}\bigg\rvert_{{\bm \alpha} = \tilde {\bm \alpha}} = \hat {\bm U}(\tilde {\bm \alpha}) = \hat{\bm U}(\hat {\bm \alpha}) + \Big(\frac{\partial\hat{\bm U}(\hat{\bm \alpha})}{\partial {\bm \alpha}}\Big) \cdot (\tilde {\bm \alpha} - \hat {\bm \alpha}) + o_p(1) = o_p(1).
\end{equation}

The first term $\ell_{\text{de}}(\tilde {\bm \alpha})$ is a constant over $\bm b$, therefore we have
\begin{equation}
\begin{aligned}
T_L &= 2\Big(\inf_{\bm b \in M} \ell_{\text{de}}(\bm b) - \inf_{\bm b \in C} \ell_{\text{de}}(\bm b)\Big) \\
&= \inf_{\bm b \in M}\Big\{ (\tilde{\bm \alpha}-\bm b)^\top  \hat H_{{\bm \alpha}|\bm \theta} (\tilde{\bm \alpha}-\bm b)+ o_p(1)\Big\} - \inf_{\bm b \in C}\Big\{ (\tilde{\bm \alpha}-\bm b)^\top  \hat H_{{\bm \alpha}|\bm \theta} (\tilde{\bm \alpha}-\bm b) + o_p(1) \Big\} \\
&= T_w + o_p(1).
\end{aligned}
\end{equation}

This shows that $T_L$ has the same asymptotic distribution as $T_w$. Similarly, for $T_s$ we have the local approximation
\begin{equation}
\hat {\bm U}(b) = \hat {\bm U}(\tilde {\bm \alpha}) + \hat H_{{\bm \alpha}|\bm \theta} \cdot (\bm b - \tilde {\bm \alpha}) + o_p(1).
\end{equation}

Plugging in $\bm b_M$ and $\bm b_C$ we obtain
\begin{equation}
\begin{aligned}
T_s &= \Big( \hat{\bm U}(\bm b_M) - \hat{\bm U}(\bm b_C) \Big)^\top  \hat H_{{\bm \alpha} | \bm \theta}^{-1} \Big( \hat{\bm U}(\bm b_M) - \hat{\bm U}(\bm b_C) \Big) \\
&= \big( \bm b_M - \bm b_C \big)^\top  \hat H_{{\bm \alpha} | \bm \theta} \big( \bm b_M - \bm b_C \big) + o_p(1) \\
&= \big( \tilde{\bm \alpha} - \bm b_M \big)^\top  \hat H_{{\bm \alpha} | \bm \theta} \big( \tilde{\bm \alpha} - \bm b_M \big) - \big( \tilde{\bm \alpha} - \bm b_C \big)^\top  \hat H_{{\bm \alpha} | \bm \theta} \big( \tilde{\bm \alpha} - \bm b_C \big) + o_p(1) \\
&= T_w + o_p(1),
\end{aligned}
\end{equation}
where the third equality comes from the Pythagoras' theorem
\begin{equation}
\big\| \tilde{\bm \alpha} - P(\tilde{\bm \alpha}, C) \big\|^2 = \big\| \tilde{\bm \alpha} - P(\tilde{\bm \alpha}, M) \big\|^2 + \big\| P(\tilde{\bm \alpha}, M) - P(\tilde{\bm \alpha}, C) \big\|^2, 
\end{equation}
and the fact that 
\begin{equation}
\begin{aligned}
\bm b_M &= \arg\inf_{\bm b \in M} \ell_{\text{de}}(\bm b) = \arg\inf_{\bm b \in M} \Big\{ \ell_{\text{de}}(\tilde {\bm \alpha}) + \frac 12 (\bm b-\tilde{\bm \alpha})^\top  \hat H_{{\bm \alpha} | \bm \theta} (\bm b-\tilde{\bm \alpha}) + o_p(1) \Big\} \\
&= \arg\inf_{\bm b \in M} \Big\{ (\bm b-\tilde{\bm \alpha})^\top  \hat H_{{\bm \alpha} | \bm \theta} (\bm b-\tilde{\bm \alpha}) \Big\} + o_p(1). 
\end{aligned}
\end{equation}

This shows that $T_s$ has the same asymptotic distribution as $T_w$, which completes the proof.

\section{Synthetic Data} \label{sec:synthetic_data}


In this section we apply our method on synthetic datasets. We consider linear model $\bm Y = \bm X\bm \beta^* + \bm \epsilon$ with $\bm \epsilon \sim N_n(\bm 0,\sigma^2\bm I_n)$, and impose different kinds of constraints on the first two variables $\bm\alpha^* = (\alpha_1^*, \alpha_2^*) = (\beta_1^*, \beta_2^*)$. 
Specifically, we consider the following three constraints.
\begin{enumerate}
\item Monotonicity constraint. We have the monotonic constraint $\alpha_1 \leq \alpha_2$, and the hypothesis we would like to test is
\begin{equation}
H_0: \alpha_1^* = \alpha_2^* \quad\text{versus}\quad H_A: \alpha_1^* < \alpha_2^*.
\end{equation}
For the experiment we set $\alpha_1^* = \alpha_2^* = -1$ and $\alpha_i^* = 0$ elsewhere.
\item Non-negativity constraint. We have the non-negative constraint $\alpha_1, \alpha_2 \geq 0$, and the hypothesis we would like to test is
\begin{equation}
H_0: \alpha_1^* = \alpha_2^* = 0 \quad\text{versus}\quad H_A: \alpha_1^* > 0 \,\,\, \text{or} \,\,\, \alpha_2^* > 0.
\end{equation}
For the experiment we set $\alpha_1^* = \alpha_2^* = 0$, $\alpha_p^* = \alpha_{p-1}^* = 1$ where $p$ is the dimension of $\bm\beta$, and $\alpha_i^* = 0$ elsewhere.
\item Sum constraint. We have the sum constraint $\alpha_1 + \alpha_2 \leq -2$, and the hypothesis we would like to test is
\begin{equation}
H_0: \alpha_1^* + \alpha_2^* = -2 \quad\text{versus}\quad H_A: \alpha_1^* + \alpha_2^* < -2.
\end{equation}
For the experiment we set $\alpha_1^* = \alpha_2^* = -1$ and $\alpha_i^* = 0$ elsewhere.
\end{enumerate}

In low dimensions we have the Least Square estimator $\hat{\bm \beta} = (\bm X^\top \bm X)^{-1}\bm X^\top \bm Y$ with $\hat{\bm \beta} \sim N\Big(\bm \beta^*, \sigma^2(\bm X^\top \bm X)^{-1}\Big)$, from which we can construct confidence interval and hypothesis testing for $\bm \beta^*$. 
In high dimension Least Square estimator is ill-conditioned so we instead calculate penalized estimator $\hat{\bm \beta}$ according to \eqref{MLE}. For example letting $P_{\lambda}$ be $L_1$ penalty we get the LASSO estimator. 
Alternatively we can get the estimator $\hat{\bm \beta}$ under constraint. For example for non-negativity constraint, 
we can get nonnegative sparse estimator directly \cite{slawski2013non}. 

In \cite{ning2017general} the authors show that our conditions in Section \ref{sec:assumption} are satisfied for linear regression so we then follow our procedure to calculate the test statistics $T_s$, $T_w$ and $T_L$. 
We set $\sigma = 1$ and we assume $\sigma$ is known. Each row of $\bm X$ is sampled from multivariate normal distribution $\bm X \sim N_p(\bm 0, \bm \Sigma)$, where $\bm \Sigma$ is a Toeplitz matrix with $\Sigma_{jk} = \rho^{|j-k|}$. The tuning parameter is set to be $\lambda = \sqrt{\log{p}/n}$ and $\lambda' = \frac{1}{2}\sqrt{\log{p}/n}$. 
We vary $\rho \in \{0.2, 0.4, 0.6, 0.8\}$, $p \in \{100, 300, 500\}$ and for each setting we generate $n=200$ samples. The averaged empirical Type I error on 500 replicates under the three different constraints are shown in Table \ref{synthetic_monotonic} - \ref{synthetic_sum}. The designed Type I error is $5\%$.

%
%
%
%
%

\begin{table}[h]
\caption{Empirical Type I error for monotonic constraint}
\setlength{\tabcolsep}{12pt}
\begin{center}
\begin{tabular}{@{}c c c c c c c}
\hline

Method &\diagbox[width=3.5em,height=2em,trim=l]{$p$}{$\rho$}  & $0.2$ & $0.4$ & $0.6$ & $0.8$ \\\hline

Score &100 & 6.4\% & 5.6\% & 5.4\% & 6.8\%  \\
 & 300 & 6.6\% & 5.4\% & 6.4\% & 5.6\% \\
 & 500 & 6.8\% & 5.4\% & 6.6\% & 6.6\% \\ \hline
 
 Wald & 100 & 5.4\% & 4.4\% & 4.8\% & 7.0\%  \\
 & 300 & 5.0\% & 3.6\% & 5.2\% & 6.2\% \\
 & 500 & 3.8\% & 3.2\% & 3.8\% & 5.2\% \\ \hline
 
 LR & 100 & 6.2\% & 4.8\% & 5.4\% & 6.4\% \\
 & 300 & 5.2\% & 5.2\% & 5.8\% & 6.4\% \\
 & 500 & 4.8\% &4.0\% &6.0\% & 5.6\% \\\hline

\end{tabular}
\end{center}
\label{synthetic_monotonic}
\end{table}%

\begin{table}[h]
\caption{Empirical Type I error for non-negative constraint}
\setlength{\tabcolsep}{12pt}
\begin{center}
\begin{tabular}{@{}c c c c c c c}
\hline

Method &\diagbox[width=3.5em,height=2em,trim=l]{$p$}{$\rho$}  & $0.2$ & $0.4$ & $0.6$ & $0.8$ \\\hline

Score &100 & 5.6\% & 6.2\% & 6.2\% & 4.8\%  \\
 & 300 & 4.6\% & 5.2\% & 5.2\% & 6.2\% \\
 & 500 & 5.4\% & 5.6\% & 5.0\% & 5.2\% \\ \hline
 
 Wald & 100 & 5.6\% & 4.8\% & 5.0\% & 4.6\%  \\
 & 300 & 5.0\% & 3.6\% & 4.0\% & 3.6\% \\
 & 500 & 3.2\% & 4.2\% & 3.4\% & 3.2\% \\ \hline
 
 LR & 100 & 6.0\% & 5.0\% & 5.8\% & 4.6\% \\
 & 300 & 3.6\% & 4.2\% & 3.6\% & 4.8\% \\
 & 500 & 3.6\% &3.4\% &4.8\% & 4.4\% \\\hline

\end{tabular}
\end{center}
\label{synthetic_nonnegative}
\end{table}%

\begin{table}[h]
\caption{Empirical Type I error for sum constraint}
\setlength{\tabcolsep}{12pt}
\begin{center}
\begin{tabular}{@{}c c c c c c c}
\hline

Method &\diagbox[width=3.5em,height=2em,trim=l]{$p$}{$\rho$}  & $0.2$ & $0.4$ & $0.6$ & $0.8$ \\\hline

Score &100 & 4.8\% & 6.2\% & 5.4\% & 5.2\%  \\
 & 300 & 4.8\% & 4.4\% & 5.8\% & 5.4\% \\
 & 500 & 4.4\% & 3.8\% & 3.6\% & 4.0\% \\ \hline
 
 Wald & 100 & 3.6\% & 5.4\% & 4.2\% & 4.0\%  \\
 & 300 & 3.8\% & 4.4\% & 3.6\% & 4.2\% \\
 & 500 & 4.0\% & 3.6\% & 4.2\% & 4.0\% \\ \hline
 
 LR & 100 & 4.2\% & 5.6\% & 4.4\% & 5.2\% \\
 & 300 & 3.8\% & 4.4\% & 4.8\% & 5.0\% \\
 & 500 & 3.4\% &4.0\% &5.2\% & 4.6\% \\\hline

\end{tabular}
\end{center}
\label{synthetic_sum}
\end{table}%

From the three tables we see that our algorithm works well for all these three constraints. We then check the power of our algorithm. 
For each constraint, we introduce a variable $margin$ that measures how much we violate the null hypothesis (i.e. how far we are away from the boundary). 
Specifically, for monotonic constraint, we set $\alpha_1^* = -1, \alpha_2^* = -1 + margin$; for non-negative constraint we set $\alpha_1^* = \alpha_2^* = margin/2$; for sum constraint, we set $\alpha_1^* = -1, \alpha_2^* = -1 - margin$. 
Intuitively, as the margin increases, the power of the test will increase. 
For all the three constraints, we compare the power of our testing procedures to the standard Wald/Score/Likelihood ratio tests where we ignore the constraint. 
For example for monotonic constraint our method tests for 
\begin{equation}
H_0: \alpha_1^* = \alpha_2^* \quad\text{versus}\quad H_A: \alpha_1^* < \alpha_2^*, 
\end{equation}
while the standard method tests for 
\begin{equation}
H_0: \alpha_1^* = \alpha_2^* \quad\text{versus}\quad H_A: \alpha_1^* \neq \alpha_2^*.
\end{equation}

We vary $margin \in \{0, 0.05, 0.1, 0.2, 0.3, 0.5, 1\}$ where $margin = 0$ corresponds to null hypothesis and others corresponds to alternative hypothesis. 
Under the alternative hypothesis, for both our method and standard method, Wald/Score/Likelihood ratio tests gives nearly identical power. Therefore we only report the mean of them. 
The comparison results on 500 replicates are shown in Table \ref{power}, and we can see that by considering the known constraint, our tests have much stronger power. 


\begin{table}[h]
\caption{Power of the tests}
\setlength{\tabcolsep}{5pt}
\begin{center}
\begin{tabular}{@{}llccccccc}
\hline
Constraint & \diagbox[width=7em,height=2.3em,trim=l]{Method}{margin} & 0 & 0.05 & 0.1 & 0.2 & 0.3 & 0.5 & 1 \\\hline
Monotonic & Our method & 0.045  &  0.061  &  0.113 &  0.211 &   0.331 &   0.597  &  0.988\\
& Standard method & 0.047  &  0.044  &  0.068  &  0.138   & 0.235  &  0.488   & 0.978\\\hline
Non-negative & Our method & 0.036  &  0.069  &  0.112 &   0.278  &  0.504 &   0.879 &   1.000\\
& Standard method & 0.039  &  0.032  &  0.047 &   0.134  &  0.323 &   0.788 &   1.000\\\hline
Sum & Our method & 0.060  &  0.169 &   0.266&    0.478  &  0.712 &   0.950  &  1.000 \\
& Standard method & 0.041 &   0.097  &  0.156  &  0.340    &0.596   & 0.922   & 1.000 \\\hline
\end{tabular}
\end{center}
\label{power}
\end{table}



\section{Real Data}\label{sec:real_data}

In this section we apply our method to two real datasets on ARCH model and information diffusion model. For both the models, we have the \emph{intrinsic} non-negative constraint on the parameters. Therefore, to provide statistical inference on the parameters, we should use constrained testing method. 

\subsection{ARCH Model}

As a first example, we consider the application of our method in financial economics, where most of the existing works focus on estimations and predictions \cite{kelly2013market, feng2018deep, lu2019expected, feng2019factor, ball2016accruals}. However, people are usually more interested in testing whether a specific factor affects the prediction results, with a focus on testing inequality constraints \cite{wolak1989testing}. 
The model we consider is the autoregressive conditional heteroscedasticity (ARCH) model introduced in \cite{engle1982autoregressive}. ARCH model is very popular in modeling financial economic time series like exchange rates, commodity prices. The main feature is that ARCH model attempts to model the variance as well. More formally, ARCH models assume the variance of the current error term to be a function of the actual sizes of the previous time periods' error terms. To introduce the model, let $\mathcal F_t$ be the information up to time $t$, $y_t$ be the dependent variable and $\bm x_t$ be exogeneous variables included in $\mathcal F_{t-1}$ ($\bm x_t$ may contain lagged dependent variables like $y_{t-1}$ and $y_{t-2}$). An ARCH model with lag length $q$ can be written as
\begin{gather}
y_t | \mathcal F_{t-1} \sim N(\bm x_t^\top \bm \beta,h_t),  \\
h_t = \alpha_0 + \alpha_1\epsilon_{t-1}^2 + ... + \alpha_q\epsilon_{t-q}^2,  \\
\epsilon_t = y_t - \bm x_t^\top \bm \beta.
\label{AR}
\end{gather}

From the definition of the model we can see that, 
if some $\alpha_i$ in \eqref{AR} is negative, then a large value for $\epsilon_{t-i}$ would lead to negative variance for $y_t$. Hence the admissible range for $\alpha_1, ..., \alpha_q$ should be $\big\{\alpha_1 \geq 0, ..., \alpha_q \geq 0\big\}$. Therefore, the testing problem should be 
\begin{equation}
H_0:\alpha_i=0 \,\,\, \text{versus} \,\,\, H_A:\alpha_i > 0, 
\end{equation}
instead of 
\begin{equation}
H_0:\alpha_i=0 \,\,\, \text{versus} \,\,\, H_A:\alpha_i \neq 0. 
\end{equation}
In this section we focus on $\alpha_1$ and test for $H_0:\alpha_1=0$ versus $H_A:\alpha_1 > 0$.

\vspace{2mm}
The data we use are the All Ordinaries Index (Australia) from January 5, 1984 to November 29, 1985, denoted as $I_t$. This index is a weighted average of the prices of selected shares in Australia which corresponds to the Dow-Jones Index in the United States.
The data are from the \emph{Australian Financial Review}. We have a total of 484 observations. The return variable $y_i$ is defined as $\log{(I_t / I_{t-1})}$, and $\bm x_t$ are the lagged dependent variables.

We estimate $\bm \alpha$ by first estimating the best fitting autoregressive model AR($q$):
\begin{equation}
y_t = a_0 + a_1y_{t-1} + ... + a_qy_{t-q} + \epsilon_{t}.
\end{equation}

We then obtain the squares of the error $\hat \epsilon^2$ and regress them on a constant and $q$ lagged values: 
\begin{equation}
\hat \epsilon_t^2 =\hat \alpha_0 + \hat \alpha_1 \hat \epsilon_{t-1}^2 + ... + \hat \alpha_q \hat \epsilon_{t-q}^2. 
\end{equation}

The estimation is based on LASSO estimator with $L_1$ penalty. We then follow our procedure to give the $p$-value. We choose $q=30$ here and it turns out that the result is not sensitive to the choice of $q$. All the three tests give $p$-value 0.41, indicating that we should not reject the null. This result is consistent with the claim in \cite{silvapulle1995score}.

\subsection{Information Diffusion}

\begin{figure}[htbp]
\begin{center}
\includegraphics[width=0.62\textwidth]{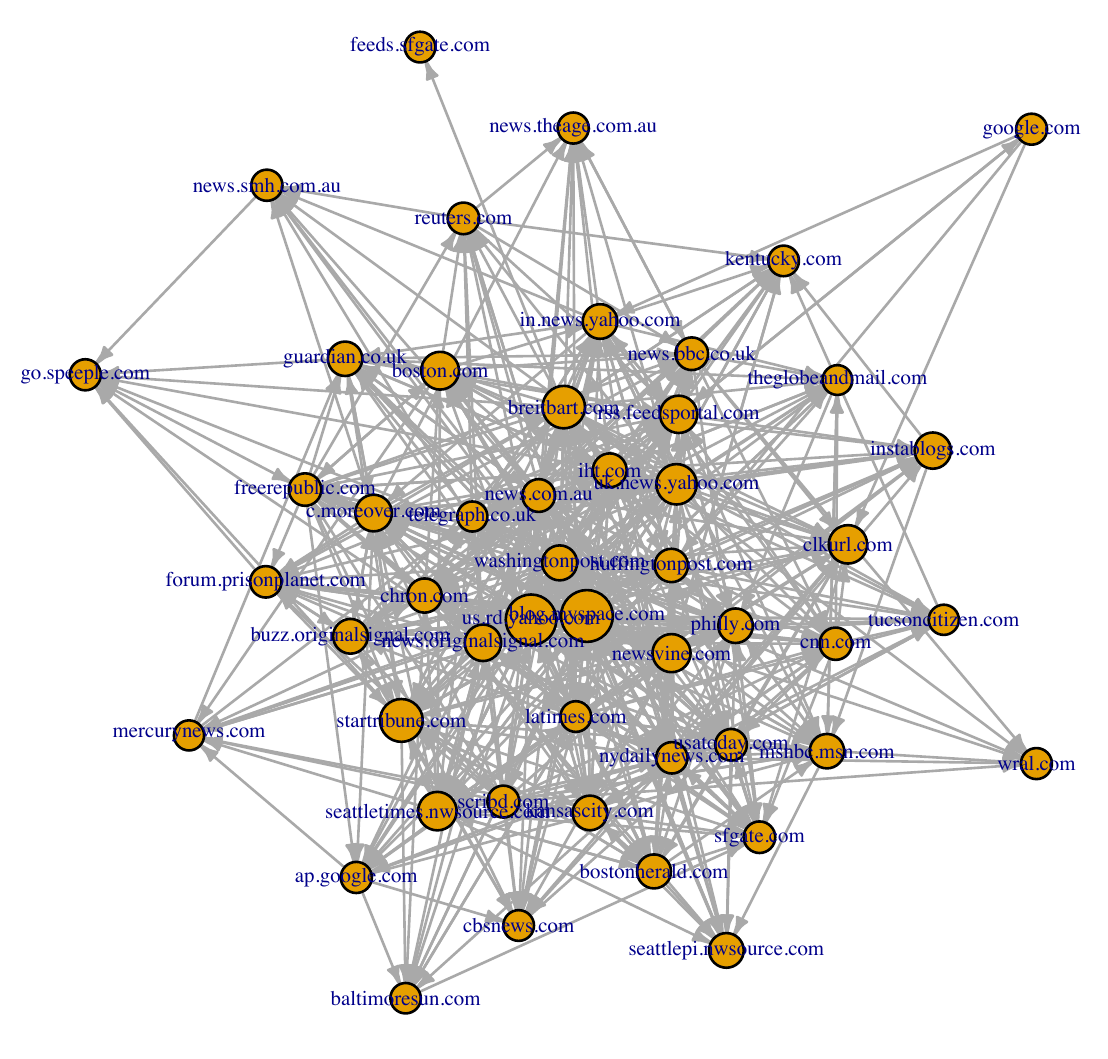}
\caption{Network structure by MLE}
\label{original}
\end{center}
\end{figure}

\begin{figure}[htbp]
\begin{center}
\includegraphics[width=0.62\textwidth]{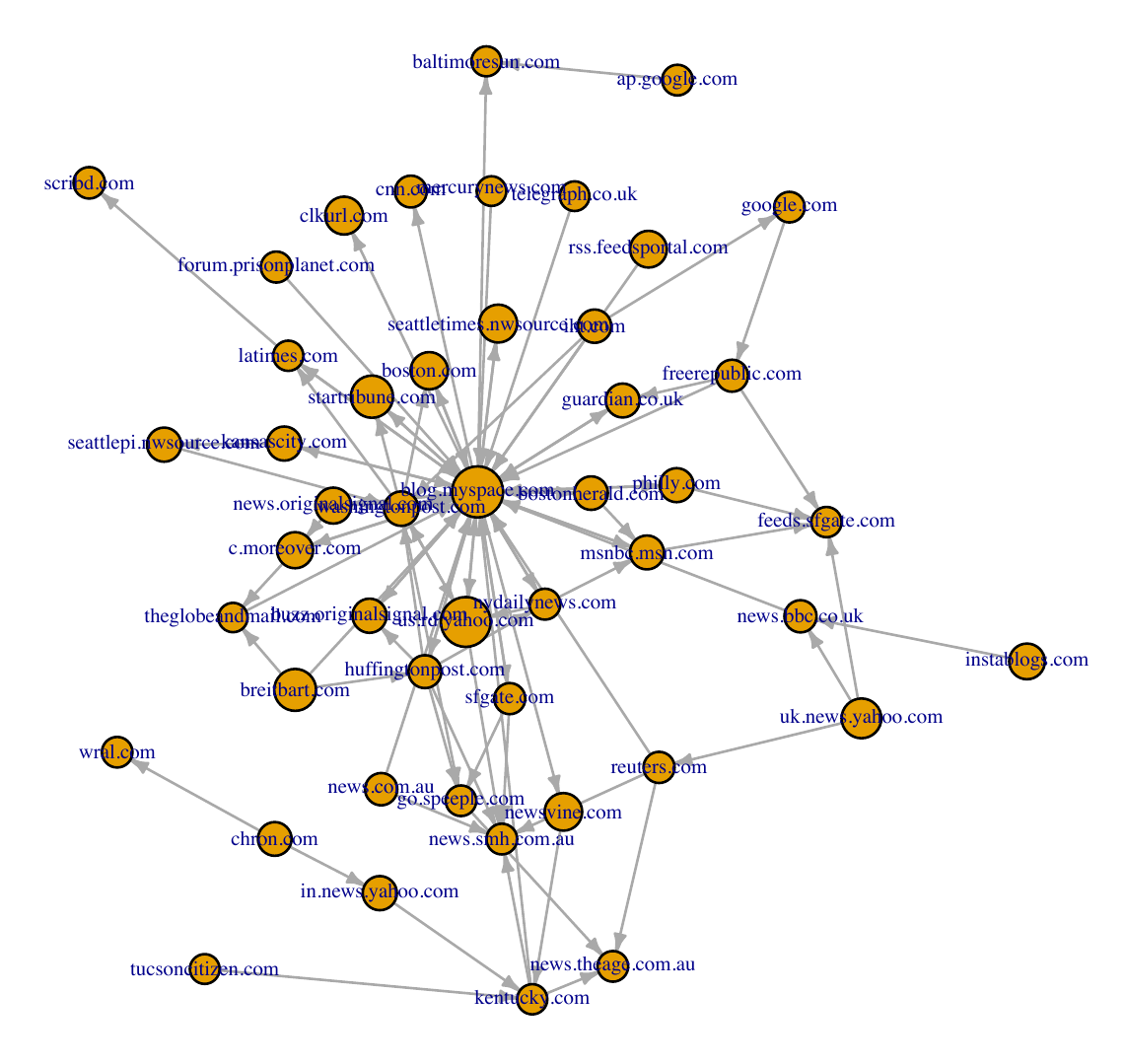}
\caption{Network structure by our algorithm with fixed critical $p$-value 0.05}
\label{ppp}
\end{center}
\end{figure}

%
%

The second model we consider is the network model. 
Network and graphical models have been widely used in fields including neuroscience, social sciences, and statistics \cite{goldenberg2010survey, zhuo2018mixing, gao2018stability, zhao2019direct, na2019estimating, drton2017structure, fortunato2010community, gao2017controllability, birge2019optimal, kim2019two}. 
We consider the time diffusion model where 
a diffusion matrix $A$ quantifies the structure between nodes. If $a_{ij} \neq 0$ then there is a link $i \to j$ and information from node $i$ may propagate to $j$. The parameter $a_{ij}$ measures how strong the relation is. Clearly only nonnegative $a_{ij}$ is meaningful in this model. Therefore if we want to know whether there exists an edge from $i$ to $j$ (i.e. whether $a_{ij} \neq 0$), this is a constrained testing problem with nonnegativity constraint (i.e. we should test for whether $a_{ij} > 0$). 
For this network diffusion problem, many existing methods \cite{rodriguez2011uncovering, gomez2015estimating, yu2018learning, Yu2017AnIM} have been proposed to recover the diffusion matrix $A$. However, all of them focus on point estimation with no statistical inference.

The specific diffusion model we use is the discretized CICE model introduced in \cite{pouget2015inferring}. We use the Memetracker dataset \cite{snapnets}\footnote{Data available at \url{http://memetracker.org}} which contains more than 172 million news articles and blog posts from 1 million online sources. This dataset contains many textual phrases (like `lipstick on a pig') extracted from websites, and the time each website mentioned it. We cluster the phrases to aggregate different textual variants of the same phrase. After aggregating
different textual variants of the same phrase, we consider each phrase cluster as a separate
cascade $c$. Since all documents are time stamped, a cascade $c$ is simply a set of time-stamps when websites first mentioned a phrase in the phrase cluster $c$. We can observe the times when websites mention a particular phrase but we don't know where they copied that phrase from. 

For the experiments we extract top 50 sites with about 2000 cascades among it. We first use penalized Maximum Likelihood Estimation for discrete CICE model in \cite{pouget2015inferring} with appropriate penalty parameter to estimate the network diffusion matrix: this network structure is shown in Figure \ref{original}. It is very dense and has many false positive edges. We then apply our algorithm to check the significance of each discovered edge. We fix the critical $p$-value to be 0.05 and keep the edges with $p$-value less than or equal to 0.05. After applying our algorithm the estimated network structure is shown in Figure \ref{ppp}. This network structure is much sparser and clearer. 
Note that this is different from using larger penalty on MLE which also gives a more sparse network structure but without statistical significance. In contrast, our procedure is able to test the significance of each edge. 
Also note that this 95\% confidence is for each edge individually, not the whole graph. If we want to recover the whole graph, that is a multiple testing problem for which we can apply multiple testing techniques on the $p$-values given by our algorithm, for example the Holm-Bonferroni method \cite{holm1979simple}.

%

\section{Conclusion}\label{sec:conclusion}

In this paper we consider the hypothesis testing problem on low dimensional parameters in high dimensional models with cone constraint on the parameters. We provide modified Wald/Score/Likelihood ratio procedures to test whether the low dimensional parameters are on the boundary of the cone constraint or not. We prove that following our procedure we can get an asymptotic designed Type I error under the null. 
Our algorithm has stronger power compared to the standard methods where we ignore the constraint.

For future work, it is of interest to consider more general constraint $C = \{ \bm\alpha: f(\bm \alpha) \geq 0 \}$ and possibly nonlinear boundary set $M$. Another future extension is to develop algorithms for models where some of our assumptions are violated. For example, for continuous time diffusion model, our Score Condition is violated \cite{gomez2015estimating}. Extending our algorithm to incorporate this model is work in progress.

\newpage
{\small
\bibliographystyle{plain}
\bibliography{paper}

\begin{thebibliography}{10}

\bibitem{andrews1998hypothesis}
Donald~WK Andrews.
\newblock Hypothesis testing with a restricted parameter space.
\newblock {\em Journal of Econometrics}, 84(1):155--199, 1998.

\bibitem{andrews2001testing}
Donald~WK Andrews.
\newblock Testing when a parameter is on the boundary of the maintained
  hypothesis.
\newblock {\em Econometrica}, pages 683--734, 2001.

\bibitem{ball2016accruals}
Ray Ball, Joseph Gerakos, Juhani~T Linnainmaa, and Valeri Nikolaev.
\newblock Accruals, cash flows, and operating profitability in the cross
  section of stock returns.
\newblock {\em Journal of Financial Economics}, 121(1):28--45, 2016.

\bibitem{belloni2014inference}
Alexandre Belloni, Victor Chernozhukov, and Christian Hansen.
\newblock Inference on treatment effects after selection among high-dimensional
  controls.
\newblock {\em The Review of Economic Studies}, 81(2):608--650, 2014.

\bibitem{berkenkamp2017safe}
Felix Berkenkamp, Matteo Turchetta, Angela Schoellig, and Andreas Krause.
\newblock Safe model-based reinforcement learning with stability guarantees.
\newblock In {\em Advances in neural information processing systems}, pages
  908--918, 2017.

\bibitem{best1990active}
Michael~J Best and Nilotpal Chakravarti.
\newblock Active set algorithms for isotonic regression; a unifying framework.
\newblock {\em Mathematical Programming}, 47(1-3):425--439, 1990.

\bibitem{birge2019optimal}
John~R Birge, Ozan Candogan, Hongfan Chen, and Daniela Saban.
\newblock Optimal commissions and subscriptions in networked markets.
\newblock {\em Forthcoming in Manufacturing \& Service Operations Management},
  2019.

\bibitem{bloch2001one}
Daniel~A Bloch, Tze~Leung Lai, and Pascale Tubert-Bitter.
\newblock One-sided tests in clinical trials with multiple endpoints.
\newblock {\em Biometrics}, 57(4):1039--1047, 2001.

\bibitem{cao2019estimation}
Jianfei Cao and Connor Dowd.
\newblock Estimation and inference for synthetic control methods with spillover
  effects.
\newblock {\em arXiv preprint arXiv:1902.07343}, 2019.

\bibitem{chernoff1954distribution}
Herman Chernoff.
\newblock On the distribution of the likelihood ratio.
\newblock {\em The Annals of Mathematical Statistics}, pages 573--578, 1954.

\bibitem{dai2019bias}
Ran Dai, Hyebin Song, Rina~Foygel Barber, and Garvesh Raskutti.
\newblock The bias of isotonic regression.
\newblock {\em arXiv preprint arXiv:1908.04462}, 2019.

\bibitem{dezeure2017high}
Ruben Dezeure, Peter B{\"u}hlmann, and Cun-Hui Zhang.
\newblock High-dimensional simultaneous inference with the bootstrap.
\newblock {\em Test}, 26(4):685--719, 2017.

\bibitem{drton2017structure}
Mathias Drton and Marloes~H Maathuis.
\newblock Structure learning in graphical modeling.
\newblock {\em Annual Review of Statistics and Its Application}, 4:365--393,
  2017.

\bibitem{engle1982autoregressive}
Robert~F Engle.
\newblock Autoregressive conditional heteroscedasticity with estimates of the
  variance of united kingdom inflation.
\newblock {\em Econometrica: Journal of the Econometric Society}, pages
  987--1007, 1982.

\bibitem{fang2017testing}
Ethan~X Fang, Yang Ning, and Han Liu.
\newblock Testing and confidence intervals for high dimensional proportional
  hazards models.
\newblock {\em Journal of the Royal Statistical Society: Series B (Statistical
  Methodology)}, 79(5):1415--1437, 2017.

\bibitem{feng2019factor}
Guanhao Feng and Jingyu He.
\newblock Factor investing: Hierarchical ensemble learning.
\newblock {\em arXiv preprint arXiv:1902.01015}, 2019.

\bibitem{feng2018deep}
Guanhao Feng, Jingyu He, and Nicholas~G Polson.
\newblock Deep learning for predicting asset returns.
\newblock {\em arXiv preprint arXiv:1804.09314}, 2018.

\bibitem{fortunato2010community}
Santo Fortunato.
\newblock Community detection in graphs.
\newblock {\em Physics reports}, 486(3-5):75--174, 2010.

\bibitem{gao2017controllability}
Zuguang Gao, Xudong Chen, and Tamer Ba{\c{s}}ar.
\newblock Controllability of conjunctive boolean networks with application to
  gene regulation.
\newblock {\em IEEE Transactions on Control of Network Systems}, 5(2):770--781,
  2017.

\bibitem{gao2018stability}
Zuguang Gao, Xudong Chen, and Tamer Ba{\c{s}}ar.
\newblock Stability structures of conjunctive boolean networks.
\newblock {\em Automatica}, 89:8--20, 2018.

\bibitem{goldenberg2010survey}
Anna Goldenberg, Alice~X Zheng, Stephen~E Fienberg, Edoardo~M Airoldi, et~al.
\newblock A survey of statistical network models.
\newblock {\em Foundations and Trends{\textregistered} in Machine Learning},
  2(2):129--233, 2010.

\bibitem{gomez2015estimating}
Manuel Gomez-Rodriguez, Le~Song, Hadi Daneshmand, and Bernhard Sch{\"o}lkopf.
\newblock Estimating diffusion networks: Recovery conditions, sample complexity
  \& soft-thresholding algorithm.
\newblock {\em Journal of Machine Learning Research}, 2015.

\bibitem{gourieroux1982likelihood}
Christian Gourieroux, Alberto Holly, and Alain Monfort.
\newblock Likelihood ratio test, wald test, and kuhn-tucker test in linear
  models with inequality constraints on the regression parameters.
\newblock {\em Econometrica: journal of the Econometric Society}, pages 63--80,
  1982.

\bibitem{hall2001order}
Daniel~B Hall and Jens~T Pr{\ae}stgaard.
\newblock Order-restricted score tests for homogeneity in generalised linear
  and nonlinear mixed models.
\newblock {\em Biometrika}, 88(3):739--751, 2001.

\bibitem{hoeffding1963probability}
Wassily Hoeffding.
\newblock Probability inequalities for sums of bounded random variables.
\newblock {\em Journal of the American statistical association},
  58(301):13--30, 1963.

\bibitem{holm1979simple}
Sture Holm.
\newblock A simple sequentially rejective multiple test procedure.
\newblock {\em Scandinavian journal of statistics}, pages 65--70, 1979.

\bibitem{asymptotic}
David~R. Hunter.
\newblock Asymptotic tools, 2011.

\bibitem{javanmard2017flexible}
Adel Javanmard and Jason~D Lee.
\newblock A flexible framework for hypothesis testing in high-dimensions.
\newblock {\em arXiv preprint arXiv:1704.07971}, 2017.

\bibitem{javanmard2014confidence}
Adel Javanmard and Andrea Montanari.
\newblock Confidence intervals and hypothesis testing for high-dimensional
  regression.
\newblock {\em The Journal of Machine Learning Research}, 15(1):2869--2909,
  2014.

\bibitem{kelly2013market}
Bryan Kelly and Seth Pruitt.
\newblock Market expectations in the cross-section of present values.
\newblock {\em The Journal of Finance}, 68(5):1721--1756, 2013.

\bibitem{kim2019two}
Byol Kim, Song Liu, and Mladen Kolar.
\newblock Two-sample inference for high-dimensional markov networks.
\newblock {\em arXiv preprint arXiv:1905.00466}, 2019.

\bibitem{king1986joint}
Maxwell~L King and Murray~D Smith.
\newblock Joint one-sided tests of linear regression coefficients.
\newblock {\em Journal of Econometrics}, 32(3):367--383, 1986.

\bibitem{king1997locally}
Maxwell~L King and Ping~X Wu.
\newblock Locally optimal one-sided tests for multiparameter hypotheses.
\newblock {\em Econometric Reviews}, 16(2):131--156, 1997.

\bibitem{knight2000asymptotics}
Keith Knight and Wenjiang Fu.
\newblock Asymptotics for lasso-type estimators.
\newblock {\em Annals of statistics}, pages 1356--1378, 2000.

\bibitem{kodde1986wald}
David~A Kodde and Franz~C Palm.
\newblock Wald criteria for jointly testing equality and inequality
  restrictions.
\newblock {\em Econometrica: journal of the Econometric Society}, pages
  1243--1248, 1986.

\bibitem{kudo1963multivariate}
Akio Kudo.
\newblock A multivariate analogue of the one-sided test.
\newblock {\em Biometrika}, pages 403--418, 1963.

\bibitem{lee2016exact}
Jason~D Lee, Dennis~L Sun, Yuekai Sun, Jonathan~E Taylor, et~al.
\newblock Exact post-selection inference, with application to the lasso.
\newblock {\em The Annals of Statistics}, 44(3):907--927, 2016.

\bibitem{snapnets}
Jure Leskovec and Andrej Krevl.
\newblock {SNAP Datasets}: {Stanford} large network dataset collection.
\newblock \url{http://snap.stanford.edu/data}, June 2014.

\bibitem{li2019statistical}
Kathleen~T Li.
\newblock Statistical inference for average treatment effects estimated by
  synthetic control methods.
\newblock {\em Journal of the American Statistical Association},
  (just-accepted):1--40, 2019.

\bibitem{lopez2019gaussian}
Andr{\'e}s~F L{\'o}pez-Lopera, ST~John, and Nicolas Durrande.
\newblock Gaussian process modulated cox processes under linear inequality
  constraints.
\newblock {\em arXiv preprint arXiv:1902.10974}, 2019.

\bibitem{lu2019expected}
Yao Lu and Valeri~V Nikolaev.
\newblock Expected loan loss provisioning: An empirical model.
\newblock {\em Chicago Booth Research Paper}, (19-11), 2019.

\bibitem{lu2013halfline}
Zeng-Hua Lu.
\newblock Halfline tests for multivariate one-sided alternatives.
\newblock {\em Computational Statistics \& Data Analysis}, 57(1):479--490,
  2013.

\bibitem{lu2015extended}
Zeng-Hua Lu.
\newblock Extended maxt tests of one-sided hypotheses.
\newblock {\em Journal of the American Statistical Association},
  (just-accepted), 2015.

\bibitem{molenberghs2007likelihood}
Geert Molenberghs and Geert Verbeke.
\newblock Likelihood ratio, score, and wald tests in a constrained parameter
  space.
\newblock {\em The American Statistician}, 61(1):22--27, 2007.

\bibitem{na2019estimating}
Sen Na, Mladen Kolar, and Oluwasanmi Koyejo.
\newblock Estimating differential latent variable graphical models with
  applications to brain connectivity.
\newblock {\em arXiv preprint arXiv:1909.05892}, 2019.

\bibitem{negahban2009unified}
Sahand Negahban, Bin Yu, Martin~J Wainwright, and Pradeep~K Ravikumar.
\newblock A unified framework for high-dimensional analysis of $ m $-estimators
  with decomposable regularizers.
\newblock In {\em Advances in Neural Information Processing Systems}, pages
  1348--1356, 2009.

\bibitem{ning2017general}
Yang Ning, Han Liu, et~al.
\newblock A general theory of hypothesis tests and confidence regions for
  sparse high dimensional models.
\newblock {\em The Annals of Statistics}, 45(1):158--195, 2017.

\bibitem{perlman1969one}
Michael~D Perlman.
\newblock One-sided testing problems in multivariate analysis.
\newblock {\em The Annals of Mathematical Statistics}, pages 549--567, 1969.

\bibitem{perlman2006some}
Michael~D Perlman and Lang Wu.
\newblock Some improved tests for multivariate one-sided hypotheses.
\newblock {\em Metrika}, 64(1):23--39, 2006.

\bibitem{pouget2015inferring}
Jean Pouget-Abadie and Thibaut Horel.
\newblock Inferring graphs from cascades: A sparse recovery framework.
\newblock In {\em Proceedings of the 24th International Conference on World
  Wide Web Companion}, pages 625--626, 2015.

\bibitem{inference}
Suba Rao.
\newblock {\em Lectures on statistical inference}.

\bibitem{ren2015asymptotic}
Zhao Ren, Tingni Sun, Cun-Hui Zhang, Harrison~H Zhou, et~al.
\newblock Asymptotic normality and optimalities in estimation of large gaussian
  graphical models.
\newblock {\em The Annals of Statistics}, 43(3):991--1026, 2015.

\bibitem{rodriguez2011uncovering}
Manuel~Gomez Rodriguez, David Balduzzi, and Bernhard Sch{\"o}lkopf.
\newblock Uncovering the temporal dynamics of diffusion networks.
\newblock {\em arXiv preprint arXiv:1105.0697}, 2011.

\bibitem{rogers1986modified}
Alan~J Rogers.
\newblock Modified lagrange multiplier tests for problems with one-sided
  alternatives.
\newblock {\em Journal of Econometrics}, 31(3):341--361, 1986.

\bibitem{shapiro1988towards}
Alexander Shapiro.
\newblock Towards a unified theory of inequality constrained testing in
  multivariate analysis.
\newblock {\em International Statistical Review/Revue Internationale de
  Statistique}, pages 49--62, 1988.

\bibitem{silvapulle2011constrained}
Mervyn~J Silvapulle and Pranab~Kumar Sen.
\newblock {\em Constrained statistical inference: Order, inequality, and shape
  constraints}, volume 912.
\newblock John Wiley \& Sons, 2011.

\bibitem{silvapulle1995score}
Mervyn~J Silvapulle and Paramsothy Silvapulle.
\newblock A score test against one-sided alternatives.
\newblock {\em Journal of the American Statistical Association},
  90(429):342--349, 1995.

\bibitem{slawski2013non}
Martin Slawski, Matthias Hein, et~al.
\newblock Non-negative least squares for high-dimensional linear models:
  Consistency and sparse recovery without regularization.
\newblock {\em Electronic Journal of Statistics}, 7:3004--3056, 2013.

\bibitem{song2019convex}
Hyebin Song, Ran Dai, Garvesh Raskutti, and Rina~Foygel Barber.
\newblock Convex and non-convex approaches for statistical inference with noisy
  labels.
\newblock {\em arXiv preprint arXiv:1910.02348}, 2019.

\bibitem{susko2013likelihood}
Edward Susko.
\newblock Likelihood ratio tests with boundary constraints using data-dependent
  degrees of freedom.
\newblock {\em Biometrika}, page ast032, 2013.

\bibitem{taylor2014exact}
Jonathan Taylor, Richard Lockhart, Ryan~J Tibshirani, and Robert Tibshirani.
\newblock Exact post-selection inference for forward stepwise and least angle
  regression.
\newblock {\em arXiv preprint arXiv:1401.3889}, 2014.

\bibitem{van2014asymptotically}
Sara Van~de Geer, Peter B{\"u}hlmann, Ya’acov Ritov, Ruben Dezeure, et~al.
\newblock On asymptotically optimal confidence regions and tests for
  high-dimensional models.
\newblock {\em The Annals of Statistics}, 42(3):1166--1202, 2014.

\bibitem{wei2019geometry}
Yuting Wei, Martin~J Wainwright, Adityanand Guntuboyina, et~al.
\newblock The geometry of hypothesis testing over convex cones: Generalized
  likelihood ratio tests and minimax radii.
\newblock {\em The Annals of Statistics}, 47(2):994--1024, 2019.

\bibitem{wen2018constrained}
Min Wen and Ufuk Topcu.
\newblock Constrained cross-entropy method for safe reinforcement learning.
\newblock In {\em Advances in Neural Information Processing Systems}, pages
  7450--7460, 2018.

\bibitem{wolak1989testing}
Frank~A Wolak.
\newblock Testing inequality constraints in linear econometric models.
\newblock {\em Journal of econometrics}, 41(2):205--235, 1989.

\bibitem{yang2019contraction}
Fan Yang, Rina~Foygel Barber, et~al.
\newblock Contraction and uniform convergence of isotonic regression.
\newblock {\em Electronic Journal of Statistics}, 13(1):646--677, 2019.

\bibitem{yang2016selective}
Fan Yang, Rina~Foygel Barber, Prateek Jain, and John Lafferty.
\newblock Selective inference for group-sparse linear models.
\newblock In {\em Advances in Neural Information Processing Systems}, pages
  2469--2477, 2016.

\bibitem{Yu2017AnIM}
Ming Yu, Varun Gupta, and Mladen Kolar.
\newblock An influence-receptivity model for topic based information cascades.
\newblock {\em 2017 IEEE International Conference on Data Mining (ICDM)}, pages
  1141--1146, 2017.

\bibitem{yu2018learning}
Ming Yu, Varun Gupta, and Mladen Kolar.
\newblock Learning influence-receptivity network structure with guarantee.
\newblock {\em arXiv preprint arXiv:1806.05730}, 2018.

\bibitem{yu2019simultaneous}
Ming Yu, Varun Gupta, and Mladen Kolar.
\newblock Simultaneous inference for pairwise graphical models with generalized
  score matching.
\newblock {\em arXiv preprint arXiv:1905.06261}, 2019.

\bibitem{yu2016statistical}
Ming Yu, Mladen Kolar, and Varun Gupta.
\newblock Statistical inference for pairwise graphical models using score
  matching.
\newblock In {\em Advances in Neural Information Processing Systems}, pages
  2829--2837, 2016.

\bibitem{yu2019convergent}
Ming Yu, Zhuoran Yang, Mladen Kolar, and Zhaoran Wang.
\newblock Convergent policy optimization for safe reinforcement learning.
\newblock {\em arXiv preprint arXiv:1910.12156}, 2019.

\bibitem{zhang2014confidence}
Cun-Hui Zhang and Stephanie~S Zhang.
\newblock Confidence intervals for low dimensional parameters in high
  dimensional linear models.
\newblock {\em Journal of the Royal Statistical Society: Series B (Statistical
  Methodology)}, 76(1):217--242, 2014.

\bibitem{zhang2017simultaneous}
Xianyang Zhang and Guang Cheng.
\newblock Simultaneous inference for high-dimensional linear models.
\newblock {\em Journal of the American Statistical Association},
  112(518):757--768, 2017.

\bibitem{zhao2019direct}
Boxin Zhao, Y~Samuel Wang, and Mladen Kolar.
\newblock Direct estimation of differential functional graphical models.
\newblock {\em arXiv preprint arXiv:1910.09701}, 2019.

\bibitem{zhu2014testing}
Rong Zhu and Sherry~ZF Zhou.
\newblock Testing inequality constraints in a linear regression model with
  spherically symmetric disturbances.
\newblock {\em Journal of Systems Science and Complexity}, 27(6):1204--1212,
  2014.

\bibitem{zhuo2018mixing}
Bumeng Zhuo and Chao Gao.
\newblock Mixing time of metropolis-hastings for bayesian community detection.
\newblock {\em arXiv preprint arXiv:1811.02612}, 2018.

\end{thebibliography}
}

\newpage
\appendix

\section{SUPPLEMENTARY MATERIAL}


We provide the proofs of Lemmas used in the paper. Some of the proofs are motivated by \cite{fang2017testing}. 

\begin{lemma}
\label{lemma:CLT}
Suppose all the conditions in Section \ref{sec:assumption} are satisfied, for any vector $\bm v \in \mathbb{R}^p$ with $\| \bm v \|_0 \leq s$, we have
\begin{equation}
\label{CLT}
\frac{\sqrt{n} \bm v^\top  \nabla \ell(\bm \beta^*)}{\sqrt{\bm v^\top \ H^*\bm v}} \overset{d}{\to} N(0,1).
\end{equation}
\end{lemma}

\begin{proof}
We define $\xi_i(\bm \beta^*) = -\nabla \log \mathcal L_i(\bm \beta^*)$, where $\mathcal L_i$ is the likelihood function for one trial $i$. According to \eqref{negative_likelihood}, we have $\nabla \ell(\bm \beta^*) = -\frac{1}{n} \sum_i \nabla \log \mathcal L_i(\bm \beta^*) = \frac{1}{n} \sum_i \xi_i(\bm \beta^*)$ and from now we write $\xi_i$ for simplicity. From the definition of $H^*$ we have 
\begin{equation}
H^* = n\text{Var}\Big(\ell(\bm \beta^*)\Big) = \text{Var}(\xi_i),
\end{equation}
and hence
\begin{equation}
\text{Var}(\bm v^\top \xi_i) = \bm v^\top H^*\bm v.
\end{equation}

From the score condition we have $\mathbb{E} [\xi_i] = 0$ and hence
\begin{equation}
\mathbb{E} [\bm v^\top \xi_i] = 0.
\end{equation}

We then know that $\frac{\bm v^\top \xi_i}{\sqrt{\bm v^\top \bm H^*\bm v}}$ has mean 0 and variance 1. Therefore the LHS of \eqref{CLT} is sum of $n$ independent random variables. We then verify the Lyapunov condition \cite{asymptotic}:
\begin{equation}
\begin{aligned}
 & \quad \,\lim_{n\to\infty} n^{-\frac{3}{2}} \sum_i \mathbb{E} \bigg|\frac{\bm{v}^\top \xi_i}{\sqrt{\bm{v}^\top H^*\bm{v}}}\bigg|^3 \\
& \leq \lim_{n\to\infty} n^{-\frac{3}{2}} \sum_i \mathbb{E} \bigg|\frac{\bm{v}^\top \xi_i}{\sqrt{c_{\min}} \|\bm{v}\|_2}\bigg|^3 \\
& \leq \lim_{n\to\infty} n^{-\frac{3}{2}} \sum_i \mathbb{E} \bigg|\frac{\bm{v}^\top \xi_i}{\sqrt{\frac{c_{\min}}{s}} \|\bm{v}\|_1}\bigg|^3 \\
& = \lim_{n\to\infty} n^{-\frac{3}{2}} \Big(\frac{s}{c_{\min}}\Big)^{\frac{3}{2}}  \sum_i \mathbb{E} \bigg|\frac{\bm{v}^\top \xi_i}{ \|\bm{v}\|_1}\bigg|^3 \\
& \leq \lim_{n\to\infty} n^{-\frac{1}{2}} \Big(\frac{s}{c_{\min}}\Big)^{\frac{3}{2}} \max{(\xi_i)} \\
& = 0,
\end{aligned}
\end{equation}
where the first inequality comes from the sparse eigenvalue condition on $H^*$ with sparse $\bm v$; the second inequality comes from Cauchy-Schwarz inequality and $\| \bm{v} \|_0 \leq s$.

Now since the Lyapunov condition is satisfied, we can apply the central limit theorem and obtain
\begin{equation}
\frac{1}{\sqrt{n}}\frac{\sum_i\bm{v}^\top \xi_i}{\sqrt{\bm{v}^\top H^*\bm{v}}}  \overset{d}{\to} N(0,1),
\end{equation}
which is just
\begin{equation}
\frac{\sqrt{n} \bm{v}^\top  \nabla \ell(\beta^*)}{\sqrt{\bm{v}^\top H^*\bm{v}}} \overset{d}{\to} N(0,1).
\end{equation}
\end{proof}


\begin{lemma} 
\label{lemma:grad_true}
Suppose all the conditions in Section \ref{sec:assumption} are satisfied, we have
\begin{equation}
\label{grad_true}
\|\nabla \ell(\bm \beta^*)\|_{\infty} = \mathcal O_{\mathbb P}\Big( \sqrt{\frac{\log{p}}{n}}\Big).
\end{equation}
\end{lemma}
\begin{proof}
Each element $\big[\nabla \ell(\bm \beta^*)\big]_j$ is the average over $n$ terms with absolute value bounded by $a$. According to Hoeffding's inequality \cite{hoeffding1963probability} we have
\begin{equation}
P\Big(\big| \big[\nabla\ell(\bm \beta^*)\big]_j \big| \geq t\Big) \leq 2e^{-\frac{nt^2}{2a^2}}.
\end{equation}

Apply union bound and let $t = C\sqrt{\frac{\log p}{n}}$ we have
\begin{equation}
\begin{aligned}
P\bigg(\|\nabla \ell(\bm \beta^*)\|_{\infty} > C\sqrt{\frac{\log p}{n}}\bigg) 
\leq p \cdot P\Big(\big| \big[\ell(\bm \beta^*)\big]_j \big| \geq C\sqrt{\frac{\log p}{n}}\Big) 
\leq p \cdot 2e^{-\frac{C^2\log p}{2a^2}} 
\leq 2p^{1-\frac{C^2}{2a^2}}.
\end{aligned}
\end{equation}

We can take large enough $C$ so that \eqref{grad_true} holds with high probability.
\end{proof}

%

\begin{lemma}
\label{lemma:hessian_true}
Suppose all the conditions in Section \ref{sec:assumption} are satisfied, we have
\begin{equation}
\label{hessian_true}
\Big\|\nabla^2\ell(\bm \beta^*) - H^*\Big\|_{\infty}  = \mathcal O_{\mathbb P}\Big( \sqrt{\frac{\log{p}}{n}}\Big).
\end{equation}
\end{lemma}

\begin{proof}

By Hoeffding's inequality again we have
\begin{equation}
\begin{aligned}
P\bigg(\big|\nabla^2_{jk}\ell(\bm \beta^*) - H^*_{jk}\big| \geq C\sqrt{\frac{\log{p}}{n}}\bigg) 
 \leq 2\exp \Big\{-\frac{2n^2C^2\frac{\log{p}}{n}}{4na^2}\Big\} 
  \leq p^{-\frac{C^2}{2a^2}}. \\
\end{aligned}
\end{equation}

Apply union bound we have
\begin{equation}
\begin{aligned}
P\bigg(\Big\|\nabla^2\ell(\bm \beta^*) - H^*\Big\|_{\infty}  \geq C\sqrt{\frac{\log{p}}{n}}\bigg) 
 \leq \sum_{j,k=1...p} P\bigg(\Big|\nabla^2_{jk}\ell(\bm \beta^*) - H^*_{jk}\Big| \geq C\sqrt{\frac{\log{p}}{n}}\bigg) 
 \leq 2p^{2-\frac{C^2}{2a^2}}.
\end{aligned}
\end{equation}

We can take large enough $C$ so that \eqref{hessian_true} holds with high probability.

\end{proof}

\begin{lemma} 
\label{lemma:H}
Suppose all the conditions in Section \ref{sec:assumption} are satisfied, for any $\widetilde{\bm \beta} = \bm \beta^* + u(\widehat{\bm \beta}-\bm \beta^*)$ with $u \in [0,1]$ we have
\begin{gather}
\| \nabla^2 \ell(\widetilde{\bm \beta}) \|_{\infty} = \mathcal{O}_{\mathbb P}(1),  \\
\|\nabla^2 \ell(\widetilde{\bm \beta})-H^*\|_{\infty}=\mathcal{O}_{\mathbb P}\Big(s\sqrt{\frac{\log{p}}{n}}\Big). 
\end{gather}
\end{lemma}
\begin{proof}

From the definition we know $\widetilde{\bm \beta}$ is of the same order with $\bm \beta^*$ and $\widehat{\bm \beta}$. The first claim comes from the second claim and the condition $\| H^* \|_{\infty} = \mathcal{O}(1)$. For the second claim, we have,
\begin{equation}
\begin{aligned}
\label{H_2}
\|\nabla^2 \ell(\widetilde{\bm \beta})-H^*\|_{\infty} &\leq \|\nabla^2 \ell(\widetilde{\bm \beta})-\nabla^2 \ell(\bm \beta^*)\|_{\infty}  + \|\nabla^2 \ell(\bm \beta^*)-H^*\|_{\infty}.
\end{aligned}
\end{equation}

For the first term in \eqref{H_2}, according to Smooth Hessian Condition and Estimation Accuracy Condition we have
\begin{equation}
\|\nabla^2 \ell(\widetilde{\bm \beta})-\nabla^2 \ell(\bm \beta^*)\|_{\infty} \leq L \cdot \|\widetilde{\bm \beta} - \bm \beta^*\|_1 = \mathcal O\Big( s\sqrt{\frac{\log{p}}{n}}\Big).
\end{equation}

For the second term in \eqref{H_2}, by Lemma \ref{lemma:hessian_true} it is $\mathcal{O}_{\mathbb P}\big(\sqrt{\frac{\log{p}}{n}}\big)$. Combining this two terms we get our desired result.

\end{proof}

\begin{lemma}

Suppose all the conditions in Section \ref{sec:assumption} are satisfied, we have
\begin{equation}
\label{C6}
\| \nabla^2_{\alpha\bm \theta}\ell(\widehat{\bm \beta}) - \bm w^{*T} \nabla^2_{\bm \theta\bm \theta}\ell(\widehat{\bm \beta}) \|_{\infty} = \mathcal{O}_{\mathbb P}\Big(s^2\sqrt{\frac{\log{p}}{n}}\Big).
\end{equation}
\end{lemma}

\begin{proof}
By triangle inequality we have
\begin{equation}
\begin{aligned}
\|\nabla^2_{\alpha\bm \theta}\ell(\widehat{\bm \beta}) - \bm w^{*T} \nabla^2_{\bm \theta\bm \theta}\ell(\widehat{\bm \beta}) \|_{\infty}  &\leq \| H_{\alpha\bm \theta}^* - \bm w^{*T}H_{\bm \theta\bm \theta}^* \|_{\infty} +  \| \nabla^2_{\alpha\bm \theta}\ell(\widehat{\bm \beta}) - H_{\alpha\bm \theta}^* \|_{\infty}  \\
& \qquad\qquad\qquad\qquad\quad + \| \bm w^{*T} \{ H_{\bm \theta\bm \theta}^* - \nabla^2_{\bm \theta\bm \theta}\ell(\widehat{\bm \beta}) \} \|_{\infty}.
\end{aligned}
\end{equation}

The first term is 0 by definition. The second term is $\mathcal{O}_{\mathbb P}\big(s\sqrt{\frac{\log{p}}{n}}\big)$ according to Lemma \ref{lemma:H}. The third term is $\mathcal{O}_{\mathbb P}\Big(s^2\sqrt{\frac{\log{p}}{n}}\Big)$ according to Lemma \ref{lemma:H} and the sparse condition $\| \bm w^* \|_1 = s$. 
Combining these three terms we get our desired result.
%

\end{proof}

\begin{lemma}
\label{lemma:diff_w}
Suppose all the conditions in Section \ref{sec:assumption} are satisfied, we have
\begin{equation}
\|\widehat{\bm w} - \bm w^*\|_1 = \mathcal{O}_{\PP}(\lambda' s) \,\,\, \text{and} \,\,\, \|\widehat{\bm w} - \bm w^*\|_2 = \mathcal{O}_{\PP}(\lambda' \sqrt s).
\end{equation}
\end{lemma}
\begin{proof}

By definition we know $\widehat{\bm w}$ satisfies
\begin{equation}
\| \nabla^2_{\alpha\bm \theta}\ell(\widehat{\bm \beta}) - \hat{\bm w}^{T} \nabla^2_{\bm \theta\bm \theta}\ell(\widehat{\bm \beta}) \|_{\infty} \leq \lambda'.
\end{equation}

Define $\bm \delta = \widehat{\bm w} - \bm w^*$, according to \eqref{C6} we have
\begin{equation}
\begin{aligned}
\| \nabla^2_{\bm \theta\bm \theta}\ell(\widehat{\bm \beta}) \cdot \bm \delta \|_{\infty} \leq \| \nabla^2_{\alpha\bm \theta}\ell(\widehat{\bm \beta}) - \widehat{\bm w}^{T} \nabla^2_{\bm \theta\bm \theta}\ell(\widehat{\bm \beta}) \|_{\infty} + \| \nabla^2_{\alpha\bm \theta}\ell(\widehat{\bm \beta}) - \bm w^{*T} \nabla^2_{\bm \theta\bm \theta}\ell(\widehat{\bm \beta}) \|_{\infty} \leq C\lambda',
\end{aligned}
\end{equation}
for some constant $C$. Therefore we have
\begin{equation}
\label{ww}
\bm \delta^\top  \cdot \nabla^2_{\bm \theta\bm \theta}\ell(\widehat{\bm \beta}) \cdot \bm \delta \leq \| \bm \delta \|_{1} \cdot \| \nabla^2_{\bm \theta\bm \theta}\ell(\widehat{\bm \beta}) \cdot \bm \delta \|_{\infty} \leq C\lambda'\| \bm \delta \|_{1}.
\end{equation}

Following Lemma 3 in \cite{yu2016statistical} we know $\|\bm{\hat w}\|_0 = cs$ for some constant $c$. 
By Sparse Eigenvalue Condition, we have 
\begin{equation}
\label{SE_appendix}
\bm \delta^\top  \nabla^2_{\bm \theta\bm \theta}\ell(\widehat{\bm \beta}) \bm \delta \geq c_{\min} \|\bm \delta\|_2^2.
\end{equation}

Plug into \eqref{ww} we obtain
\begin{equation}
C\lambda'\| \bm \delta \|_{1} \geq c_{\min}\|\bm \delta\|_2^2 \geq c_{\min}\|\bm \delta\|_1^2\cdot \frac{1}{s},
\end{equation}
which gives
\begin{equation}
\|\bm \delta\|_1 \leq \frac{C\lambda's}{c_{\min}} = \mathcal{O}_{\PP}(\lambda's),
\end{equation}
and also
\begin{equation}
\|\bm \delta\|_2  = \mathcal{O}_{\PP}(\lambda' \sqrt{s}),
\end{equation}
\end{proof}

\begin{remark}
\label{RE_OK}
We show that Restricted Eigenvalue Condition also works here, as discussed in Remark \ref{SE_RE}. According to the optimality condition of Dantzig selector we have $\|\bm{\hat w}\|_1 \leq \|\bm{w^*}\|_1$. 
Also note that since $\|\bm w^*_{\mathcal S^c}\|_1 = 0$ we have
\begin{equation}
\|\bm{\hat w}_{\mathcal S}\|_1 + \|\bm{\hat w}_{\mathcal S^c}\|_1 \leq \|\bm w^*_{\mathcal S}\|_1.
\end{equation}

By triangle inequality we have
\begin{equation}
\|\bm w^*_{\mathcal S}\|_1 \leq \|\bm{\hat w}_{\mathcal S}\|_1 + \|\bm{\delta}_{\mathcal S}\|_1.
\end{equation}

Summing up these two inequalities we obtain
\begin{equation}
\|\bm{\delta}_{\mathcal S^c}\|_1 \leq \|\bm{\delta}_{\mathcal S}\|_1,
\end{equation}
which means $\bm \delta \in \mathcal C(\mathcal S)$. Therefore with Restricted Eigenvalue Condition we can still get \eqref{SE_appendix} and everything follows.

Moreover, in the proof of Lemma \ref{lemma_normal}, since we take $\bm v = (1; -{\bm w}^*)$, clearly we have $\bm v \in \mathcal C(\mathcal S)$. Therefore the proof of Lemma \ref{lemma:CLT} also hold under Restricted Eigenvalue Condition. Combining these two results we see that Restricted Eigenvalue Condition also suffices for our algorithm to be valid.
\end{remark}

\end{document}